%% file: main.tex
\let\origvec\vec
\documentclass{llncs}

\let\vec\origvec
\usepackage{amsmath}
\usepackage{epsfig}
 
\usepackage[T1]{fontenc}
\usepackage[utf8]{inputenc}
\usepackage{amssymb}
\usepackage{graphicx}
\usepackage{bbm}
\usepackage{subfigure}
\usepackage{ifthen}
\usepackage{url}
\usepackage{algorithm}
\usepackage{algorithmic}
\usepackage{tabularx}
\usepackage{xcolor}
\usepackage{wrapfig}
\newcommand{\DB}{\ensuremath{\mathcal{D}}}
\newcommand{\RR}{\ensuremath{\mathbb{R}}}
\newcommand{\NN}{\ensuremath{\mathbb{N}}}

\title{Uncertain Spatial Data Management:\\ An Overview}

\begin{document}

\author{Andreas Z\"ufle}

\institute{
George Mason University, Fairfax, VA, USA\\
Department of Geography and Geoinformation Science\\
\email{azufle@gmu.edu}\\
{\color{red}{Preprint: To Appear in Big Geospatial Data. Chapter 3.2. Springer Books.}}
}

\maketitle

{\abstract{Both the current trends in technology such as smart phones, general mobile devices, stationary sensors, and satellites as we as a new user mentality of using this technology to voluntarily share enriched location information produces a flood of geo-spatial and geo-spatio-temporal data. This data flood provides a tremendous potential of discovering new and useful knowledge. But in addition to the fact that measurements are imprecise, spatial data is often interpolated between discrete observations. To reduce communication and bandwidth utilization, data is often subjected to a reduction, thereby eliminating some of the known/recorded values. These issues introduce the notion of uncertainty in the context of spatio-temporal data management, an aspect raising imminent need for scalable and flexible solutions. The main scope of this chapter is to survey existing techniques for managing, querying, and mining uncertain spatio-temporal data. First, this chapter surveys common data representations for uncertain data, explains the commonly used possible worlds semantics to interpret an uncertain database, and surveys existing system to process uncertain data. Then this chapter defines the notion of different probabilistic result semantics to distinguish the task of enrich individual objects with probabilities rather than enriched entire results with probabilities. To distinguish between result semantics is important, as for many queries, the problem of computing object-level result probabilities can be done efficiently, whereas the problem of computing probabilities of entire results is often exponentially hard. Then, this chapter provides an overview over probabilistic query predicates to quantify the required probability of a result to be included in the result.
Finally, this chapter introduces a novel paradigm to efficiently answer any kind of query on uncertain data: the Paradigm of Equivalent Worlds, which groups the exponential set of possible database worlds into a polynomial number of set of equivalent worlds that can be processed efficiently. Examples and use-cases of querying uncertain spatial data are provided using the example of uncertain range queries. 
}}
\clearpage

\section{Introduction}\vspace{-0.2cm}

Due to the proliferation of handheld GPS enabled devices, spatial and spatio-temporal data is generated, stored, and published by billions of users in a plethora of applications. By mining this data, and thus turning it into actionable information, The McKinsey Global Institute projects a ``\$600 billion potential annual consumer surplus from using personal location data globally''.

As the volume, variety and velocity of spatial data has increased sharply over the last decades, uncertainty has increased as well. Until the early 21st century, spatial data available for geographic information science (GIS) was mainly collected, curated, standardized \cite{fegeas1992overview,FGDC}, and published by authoritative sources such as the United States Geological Survey (USGS) \cite{USGS}. Now, data used for spatial data mining is often obtained from sources of volunteered geographic information (VGI) \cite{sui2012crowdsourcing,OSM}.
Consequentially, our ability to unearth valuable knowledge from large sets of such spatial data is often impaired by the uncertainty of the data which geography has been named the ``the Achilles heel of GIS'' \cite{goodchild1998uncertainty} for many reasons:
\begin{itemize}
\item Imprecision is caused by physical limitations of sensing devices and connection errors, for instance in geographic information system using cell-phone GPS \cite{couclelis2003certainty},\vspace{-0.0cm}
\item Data records may be obsolete. In geo-social networks and microblogging platforms such as Twitter, users may update their location infrequently, yielding uncertain location information in-between data records \cite{kumar2014twitter},\vspace{-0.0cm}
\item Data can be obtained from unreliable sources, such as volunteered geographic information like data in Open-Street-Map \cite{OSM}, where data is obtained from individual users, which may incur inaccurate or plain wrong data, deliberately or due to human error \cite{grira2010spatial},\vspace{-0.0cm}
\item Data sets pertaining to specific questions may be too small to answer questions reliably. Proper statistical inference is required to draw significant conclusions from the data and to avoid basing decisions upon spurious mining results \cite{hsu1996multiple,casella2002statistical}.\vspace{-0.0cm}
\end{itemize}

\begin{figure}[t]
    \centering
    \includegraphics[width=0.8\textwidth]{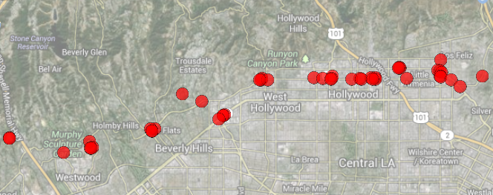}\vspace{-0.2cm}
    \caption{User locations in a Location-based social network (Gowalla) over a day. }\vspace{-0.4cm}
    \label{fig:uncertainty}
\end{figure}

\noindent To illustrate uncertainty in spatial and spatio-temporal data, Figure~\ref{fig:uncertainty} shows a typical one-day ``trajectory'' of a prolific user in the location-based social network Gowalla (data taken from \cite{cho2011friendship}). While a trajectory is usually defined as a function that continuously maps time to locations, we see that in this case, we can only observe the user at discrete times, having hours in-between subsequent location updates. Where was the user located in-between these updates? Should we use dead reckoning techniques to interpolate the locations or should be assume that the user stays at a location until next update? Also, users may spoof their location~\cite{zhao2017true}, either to protect their privacy or to gain advantages within the location-based social network. Given this uncertainty, how certain can we be about the location of the user at a given time $t$? And how does the uncertainty increase as location updates become more sparse and obsolete?
The goal of this chapter is to provide a comprehensive overview of models and techniques to deal with uncertainty. 
To handle uncertainty, we must first remind ourselves that a database models an aspect of the real world, the universe of discourse. Information observed and stored in a database may deviate from the real-world. For reliable decision making, we need to quantify the uncertainty of attribute values stored in the database and consider potentially missing objects that may change mining results.
%
%

\begin{figure}[t]
    \centering
    \includegraphics[width=0.8\textwidth]{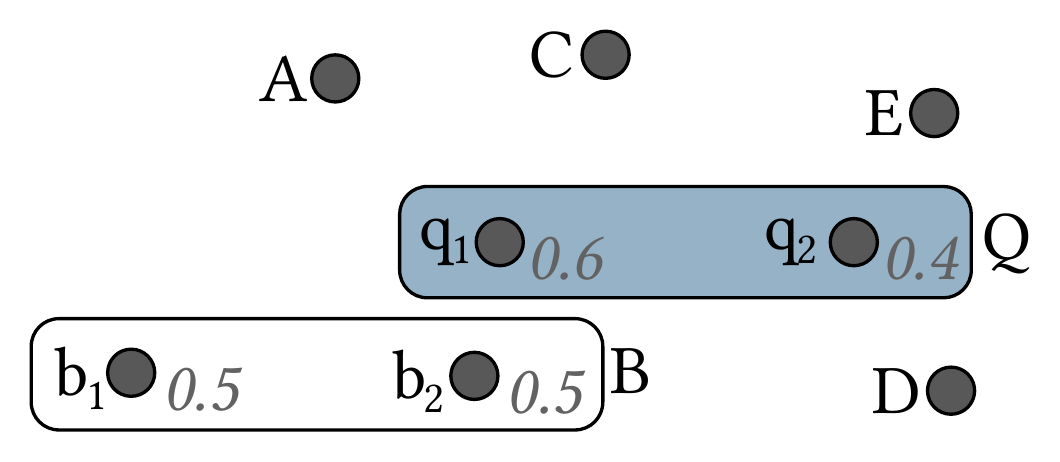}\vspace{-0.3cm}
    \caption{Exemplary Uncertain Database.}\vspace{-0.4cm}
    \label{fig:toy}
\end{figure}

\begin{example}
As a running example used through this chapter, consider Figure~\ref{fig:toy} which shows a toy uncertain spatial database. In this example, two objects, $Q$ and $B$ have uncertain locations, indicated by alternative locations $\{q_1, q_2\}$ of $Q$ and alternative locations $\{b_1,b_2\}$ of $B$. 
In this book chapter, we will survey methods to answer questions such as ``What object is closest to Q?'', or ``What is the probability of $B$ to be one of the two-nearest neighbors of $Q$?''
\end{example}
To answer such queries, we first need a crisp definition of what it means for an uncertain object to be a (probabilistic) nearest neighbor of a query object and how the probability of such an event is defined. This chapter gives a widely used interpretation of uncertain databases using \emph{Possible Worlds Semantics}. This interpretation allows to answer arbitrary queries on uncertain data, but at a computational cost exponential in the number of uncertain objects. For efficient processing, this chapter defines a paradigm of querying uncertain data that allows to efficiently answer many spatial queries on uncertain spatial data. 

managing and querying uncertain spatial data.
Parts of this section have been
presented in the form of presentation slides at recent conference tutorials at VLDB 2010 (\cite{RenCheKriZueetal10}), ICDE 2014 (\cite{cheng2014managing}), ICDE 2017 (\cite{zufle2017handling}), and MDM 2020~\cite{zufle2020managing}. This chapter is
subdivided to give a survey of
definitions, notions and techniques used in the field of querying
and mining uncertain spatio-temporal data.
\begin{itemize}
\item {\bf Section
\ref{subsec:probmodels}} presents a survey of state-of-the-art
\emph{data representations models} used in the field of uncertain
data management. This section explain discrete and continuous models for uncertain objects. 
\item To interpret queries on a database of uncertain objects,
well-defined semantics of uncertain database are required. For this purpose,
{\bf Section \ref{sec:pws}} introduces the \emph{possible world
semantics} for uncertain data.
\item To run queries on uncertain spatial data, existing systems for uncertain spatial database management are surveyed in {\bf Section~\ref{sec:systems}}.
\item Given an uncertain database,
the result of a probabilistic query can be interpreted in two ways
as elaborated in {\bf Section \ref{subsec:ProbAnswerSem}}. This distinction between different \emph{probabilistic result semantics} is not made explicitly in any related work, but is required to gain a deep understanding of problems in the field of querying uncertain spatial data and their complexity. 

\item {\bf Section~\ref{sec:probpred}} gives an overview over \emph{probabilistic
query predicates}. A probabilistic query predicate defines the
requirements for the probability of a candidate result to be returned as a query result. 

\item {\bf Section \ref{part:paradigm}} introduces a novel paradigm for
uncertain data to efficiently answer any kind of query using
possible world semantics. This \emph{Paradigm of Equivalent
Worlds} generalizes existing solutions by identifying requirements
a query must satisfy in order to have a polynomial solution. 
\item {\bf Section
\ref{chap:sumofindependent}} presents efficient solutions for the
problem of computing range queries on uncertain spatial databases. For this purpose, the paradigm of equivalent worlds is leveraged to compute the distribution of the sum of a Poisson-binomial distributed
random variable, a problem that is paramount for many spatial queries on uncertain data. 
\item {Section~\ref{ref:overview}} gives an overview of specific research problems using uncertain spatial and spatio-temporal data, and surveys state-of-the-art solutions.
\item Finally, {\bf Section~\ref{sec:summary}} concludes this book chapter and sketches future research directions that can be opened by leveraging the Paradigm of Equivalent Worlds to new applications and query types.
\end{itemize}
\clearpage

\input{kap_02}
\input{kap_03}

\bibliographystyle{acm}
\bibliography{abbreviations,bibliography,literature}

\end{document}

%% file: kap_02.tex
\begin{figure}[t]
    \centering
    \subfigure[Discrete Probability Mass Function]{
        \label{fig:disc}
        \includegraphics[width = 0.47\columnwidth]{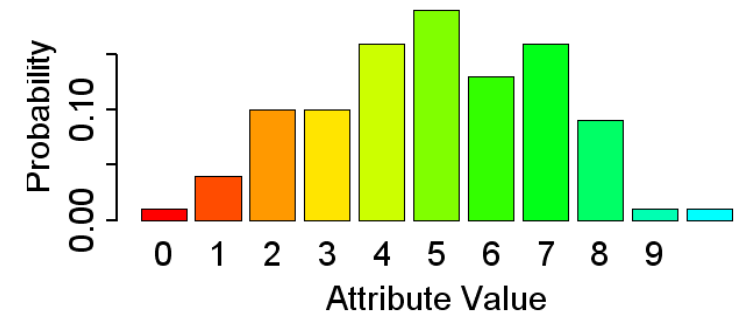}
    }
    \subfigure[Continuous Prob. Density Function]{
        \label{fig:cont}
        \includegraphics[width = 0.47   \columnwidth]{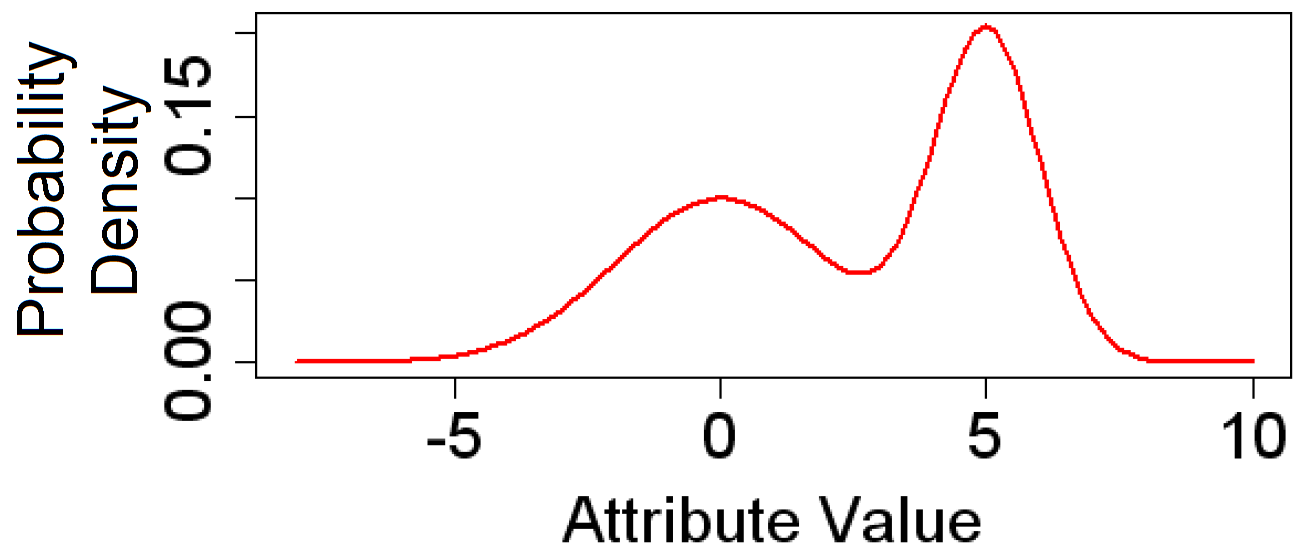}
    }\vspace{-0.2cm}
    \caption{Models for Uncertain Attributes}\vspace{-0.2cm}
\end{figure}

\section{Discrete and Continuous Models for Uncertain Data}\label{subsec:probmodels}\vspace{-0.3cm}
An object is uncertain if at least one attribute of $o$ is
uncertain. The uncertainty of an attribute can be captured in a
discrete or continuous way. A discrete model uses a probability
mass function (pmf) to describe the location of an uncertain
object. In essence, such a model describes an uncertain object by
a finite number of alternative instances, each with an associated
probability \cite{KriKunRen07,PeiHuaTaoLin08}, as shown in Figure
\ref{fig:disc}. In contrast, a continuous model uses a continuous
probability density function (pdf), like Gaussian, uniform,
Zipfian, or a mixture model, as depicted in Figure \ref{fig:cont},
to represent object locations over the space. Thus, in a
continuous model, the number of possible attribute values is
uncountably infinite. In order to estimate the probability that an
uncertain attribute value is within an interval, integration of
its pdf over this interval is required \cite{TaoCheXiaNgaetal05}.
The random variables corresponding to each uncertain attribute of
an object $o$ can be arbitrarily correlated.

To capture positional uncertainty, such models can be applied by
treating longitude and latitude (and optionally elevation) as two
(three) uncertain attributes. In the case of discrete positional
 uncertainty, the position of an object $A$ is given by a discrete
set $a_1,...,a_m$ of $m\in \NN$ possible alternatives in space, as
exemplarily depicted in Figure \ref{fig:disc_obj} for two
uncertain objects $A$ and $B$. Each alternative $a_i$ is
associated with a probability value $p(a_i)$, which may for
example be derived from empirical information about the turn
probabilities of intersection in an underlying road network. In a
nutshell, the position $A$ is a random variable, defined by a
probability mass function $\mbox{pdf}_A$ that maps each
alternative position $a_i$ to its corresponding probability
$p(a_i)$, and that maps all other positions in space to a zero
probability. An important property of uncertain spatial databases
is the inherent correlation of spatial attributes. In the example
shown in Figure \ref{fig:disc_obj} it can be observed that the
uncertain attributes $a$ and $b$ are highly correlated: given the
value of one attribute, the other attribute is certain, as there
is no two alternatives of objects $A$ and $B$ having identical
attribute values in either attribute.

Clearly, it must hold that the sum of probabilities of all
alternatives must sum to at most one:
$$
\sum_{i=1}^m p(a_i)\leq 1
$$
In the case where $\sum_{i=1}^m p(a_i)\leq 1$ object $A$ has a
non-zero probability of $1-\sum_{i=1}^m p(a_i)\geq 0$ to not exist
at all. This case is called \emph{existential uncertainty}, and
$A$ is denoted as \emph{existentially uncertain} \cite{YMDTV09}.
If the total number of possible instances $m$ is greater than one,
$A$ is denoted as \emph{attribute uncertain}. In the context of
uncertain spatial data, attribute uncertainty is also referred to
as \emph{positional uncertainty} or \emph{location uncertainty}.
An object can be both existentially uncertain and attribute
uncertain. In Figure \ref{fig:disc_obj}, object $A$ is both
existentially uncertain and attribute uncertain, while object $B$
is attribute uncertain but does exist for certain.

\begin{figure}[t]
    \centering
    \subfigure[Discrete Case]{
        \label{fig:disc_obj}
        \includegraphics[width = 0.45\columnwidth]{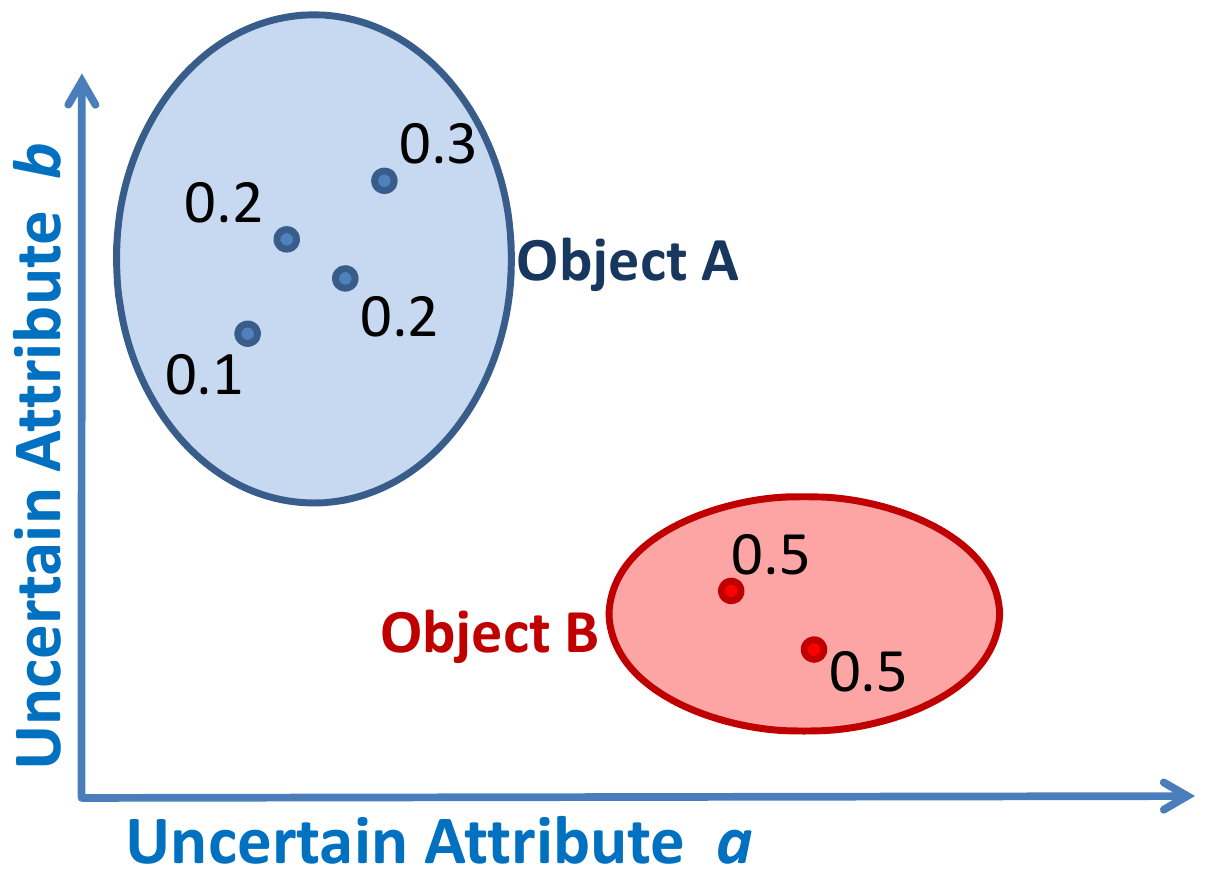}
    }
    \subfigure[Continuous Case]{
        \label{fig:cont_obj}
        \includegraphics[width = 0.45\columnwidth]{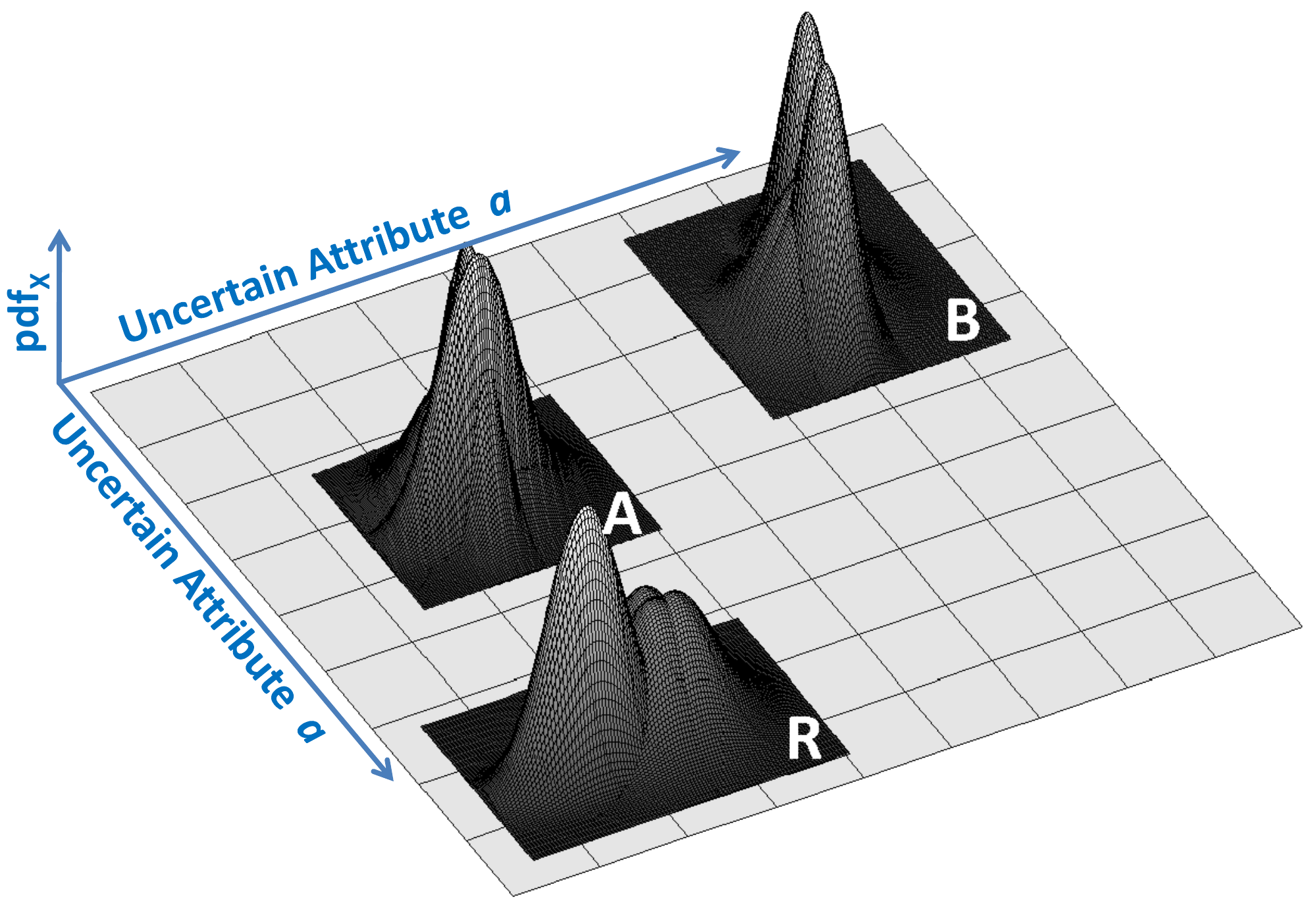}
    }
    \caption{Uncertain Objects}
\end{figure}

In the case of continuous uncertainty, the number of possible
alternative positions of an object $A$ is infinite, and given by
the non-zero domain of the probability density function
$\mbox{pdf}_x$. The probability of $A$ to occur in some spatial
region $r$ is given by integration
$$
\int_r pdf_A(x)dx.
$$
Since arbitrary pdfs may be represented by an uncountably infinite large
number of ($position$, $probability$) pairs, such pdfs may require
infinite space to represent. For this reason, assumptions on the
shape of a pdf are made in practice. All continuous models for
positionally uncertain data therefore use parametric pdfs, such as
Gaussian, uniform, Zipfian, mixture models, or parametric
spline representations. For illustration purpose, Figure
\ref{fig:cont_obj} depicts three uncertain objects modelled by a
mixture of gaussian pdfs. Similar to the discrete case, the
constraint
$$
\int_{\RR^d} pdf_A(x) dx\leq 1
$$
must be satisfied, where $\RR^d$ is a $d$ dimensional vector
space. In the case of spatial data, $d$ usually equals two or
three. The notion of existentially and attribute uncertain objects
is defined analogous to the discrete case.

 The following section reviews related work and state-of-the-art on
the field of modeling uncertain data.

\subsection{Existing Models for Uncertain
Data}\label{existing_models} This section gives a brief survey on
existing models for uncertain spatial data used in the database
community. Many of the presented models have been developed to
model uncertainty in relational data, but can be easily adapted to
model uncertain spatial data. Since one of the main challenges of
modeling uncertain data is to capture correlation between
uncertain objects, this section will elaborate details on how
state-of-the-art approaches tackles this challenge. Both discrete
and continuous models are presented.

\subsection*{Discrete Models} In addition to reviewing related
work defining discrete uncertainty models, the aim of this section
is to put these papers into context of Section
\ref{subsec:probmodels}. In particular, models which are special
cases or equivalent to the model presented in Section
\ref{subsec:probmodels} will be identified, and proper mappings to
Section \ref{subsec:probmodels} will be given.

{\bf Independent Tuple Model.} Initial models have been proposed
simultaneously and independently in \cite{Fuhr97,Zimanyi97}.
These works assume a relational model in which each tuple is
associated with a probability describing its existential
uncertainty. All tuples are considered independent from each other.
This simple model can be seen as a special case of the model
presented in Section \ref{subsec:probmodels}, where only
existential uncertain but no attribute uncertainty is modelled.

{\bf Block-Independent Disjoint Tuples Model and X-Tuple model} A
more recent and the currently most prominent approach to model
discrete uncertainty is the block-independent disjoint tuples
model (\cite{Dalvi09}), which can capture mutual exclusion between
tuples in uncertain relational databases. A probabilistic database
is called block independent-disjoint if the set of all possible
tuples can be partitioned into blocks such that tuples from the
same block are disjoint events, and tuples from distinct blocks
are independent. A commonly used example of a block-independent
disjoint tuples model is the \emph{Uncertainty-Lineage Database
Model}(\cite{Benjelloun06,Sarma06,Soliman07,YiLKS08,YiLiKolSri08a}), also called \emph{X-Relation Model} or
simply \emph{X-Tuple Model} that has been developed for relational
data. In this model, a probabilistic database is a finite set of
\emph{probabilistic tables}. A probabilistic table $T$ contains a
set of (uncertain) tuples, where each tuple $t\in T$ is associated
with a membership probability value $Pr(t) > 0$. A
\emph{generation rule} $R$ on a table $T$ specifies a set of
mutually exclusive tuples in the form of $R : t_{r_1} \oplus ...
\oplus t_{r_m}$ where $t_{r_i} \in T (1\leq i\leq m)$ and
$P(R):=\sum_{i=1}^m t_{r_i}\leq 1$. The rule R constrains that,
among all tuples $t_{r_1},...,t_{r_m}$ involved in the rule, at
most one tuple can appear in a possible world. The case where
$P(R)< 1$ the probability $1-P(R)$ corresponds to the probability
that no tuple contained in rule $R$ exists. It is assumed that for
any two rules $R_1$ and $R_2$ it holds that $R_1$ and $R_2$ do not
share any common tuples, i.e., $R_1\cap R_2=\emptyset$. In this
model, a possible world $w$ is a subset of $T$ such that for each
generation rule $R$, $w$ contains exactly one tuple involved in
$R$ if $P(R) = 1$, or $w$ contains 0 or 1 tuple involved in $R$ if
$Pr(R) < 1$.

This model can be translated to a discrete model for uncertain
spatial data as discussed in Section \ref{subsec:probmodels} by
interpreting the set $T$ as the set of all possible locations of
all objects, and interpreting each rule $R$ as an uncertain
spatial object having alternatives $t_{r_i}$. The constraint that
no two rules may share any common tuples translates into the
assumption of mutually independent spatial objects. Finally, the
case $P(R)<1$ corresponds to the case of existential uncertainty
(see Section \ref{subsec:probmodels}).

A similar block-independent disjoint tuples model is called
\emph{p-or-set} \cite{Re06} and can be translated to the model
described in Section \ref{subsec:probmodels} analogously. In
\cite{AntJanKocOlt07}, another model for uncertainty in relational
databases has been proposed that allows to represent attribute
values by sets of possible values instead of single deterministic
values. This work extends relational algebra by an operator for
computing possible results. A normalized representation of
uncertain attributes, which essentially splits each uncertain
attribute into a single relation, a so-called U-relation, allows
to efficiently answer projection-selection-join queries. The main
drawback of this model is that it is not possible to compute
probabilities of the returned possible results. Sen and Deshpande
\cite{SenDes07} propose a model based on a probabilistic graphical
model, for explicitly modeling correlations among tuples in a
probabilistic database. Strategies for executing SQL queries over
such data have been developed in this work. The main drawback of
using the proposed graphical model is its complexity, which grows
exponential in the number of mutually correlated tuples. This is a
general drawback for graphical models such as Bayesian networks
and graphical Markov models, where even a \emph{factorized
representation} may fail to reduce the complexity sufficiently:
The idea of a factorized representation is to identify conditional
independencies. For example, if a random variable $C$ depends on
random variables $A$ and $B$, then the distribution of $C$ has to
be given relative to all combination of realizations of $A$ and
$B$. If however, $C$ is conditionally independent of $A$, i.e.,
$B$ depends on $A$, $C$ depends on $B$, and $C$ only transitively
depends on $A$, then it is sufficient to store the distribution of
$C$ relative only to the realizations of $B$. Nevertheless, if for
a given graphical model a random variable depends on more than a
hand-full of other random variables, then the corresponding model
will become infeasible.

{\bf And/Xor Tree Model.} A very recent work by Li and Deshpande
\cite{LiD09} extends the block-independent disjoint tuples model
by adding support for mutual co-existence. Two events satisfy the
mutual co-existence correlation if in any possible world, either
both happen or neither occurs. This work allows both mutual
exclusiveness and mutual co-existence to be specified in a
hierarchical manner. The resulting tree structure is called an
\emph{and/xor tree}. While theoretically highly relevant, the
and/xor tree model becomes impracticable in large database having
non-trivial object dependencies, as it grows exponentially in the
number of database objects.

If not stated otherwise, this chapter will apply the
block-independent disjoint tuples model as model of choice for
discrete uncertain data.

\clearpage
\subsection*{Continuous Models} In general, similarity search
methods based on continuous models involve expensive integrations
of the PDFs, hence special approximation and indexing techniques
for efficient query processing are typically employed
\cite{CheXiaPraShaVit04,TaoCheXiaNgaetal05}. In order to increase
quality of approximations, and in order to reduce the
computational complexity, a number of models have been proposed
making assumptions on the shape of object PDFs. Such assumptions
can often be made in applications where the uncertain values
follow a specific parametric distribution, e.g. a uniform
distribution \cite{ChengKP03,ChengCMC08} or a Gaussian
distribution \cite{ChengCMC08,DeshpandeGMHH04,PatroumpasPS12}.
Multiple such distributions can be mixed to obtain a mixture model
\cite{Tran10,BoePrySch06}. To approximate arbitrary PDFs,
\cite{LiD10} proposes to use polynomial spline approximations.

\section{Possible World Semantics}\label{sec:pws}

In an uncertain spatial database $\DB=\{U_1,...,U_N\}$, the
location of an object is a random variable. Consequently, if there
is at least one uncertain object, the data stored in the database
becomes a random variable. To interpret, that is, to define the
semantics of a database that is, in itself, a random variable, the
concept of \emph{possible worlds} is described in this section.

\begin{definition}[Possible World Semantics]\label{ref:pws}
A possible world $w=\{u_1^{a_1},...,u_N^{a_N}\}$ is a set of
instances containing at most one instance $u_i^{a_i}\in U_i$ from
each object $U_i\in\DB$. The set of all possible worlds is denoted
as $\mathcal{W}$. The total probability of an uncertain world
$P(w\in\mathcal{W})$ is derived from the chain rule of conditional
probabilities:
\begin{equation}\label{eq:pw1}
P(w):=P(\bigwedge_{u_i^{a_i}\in w}U_i=u_i^{a_i})=\prod_{i=1}^{N}
P(u_i^{a_i}|\bigwedge_{j<i}u_j^{a_j}).
\end{equation}
By definition, all worlds $w$ having a zero probability $P(w)=0$
are excluded from the set of possible worlds $\mathcal{W}$.
Equation \ref{eq:pw1} can be used if conditional probabilities of
the position of objects given the position of other objects are
known, e.g. by a given graphical model such as a Bayesian network
or a Markov model. In many applications where independence between
object locations can be assumed, as well as in applications where
only the marginal probabilities $P(u_i^{a_i})$ are known, and thus
independence has to be assumed due to lack of better knowledge of
a dependency model, the above equation simplifies to
\begin{equation}\label{eq:pw2}
P(w)=\prod_{i=1}^{N} P(u_i^{a_i}).
\end{equation}
\end{definition}

\clearpage
\begin{figure}[t]
    \centering
    \includegraphics[width=0.85\columnwidth]{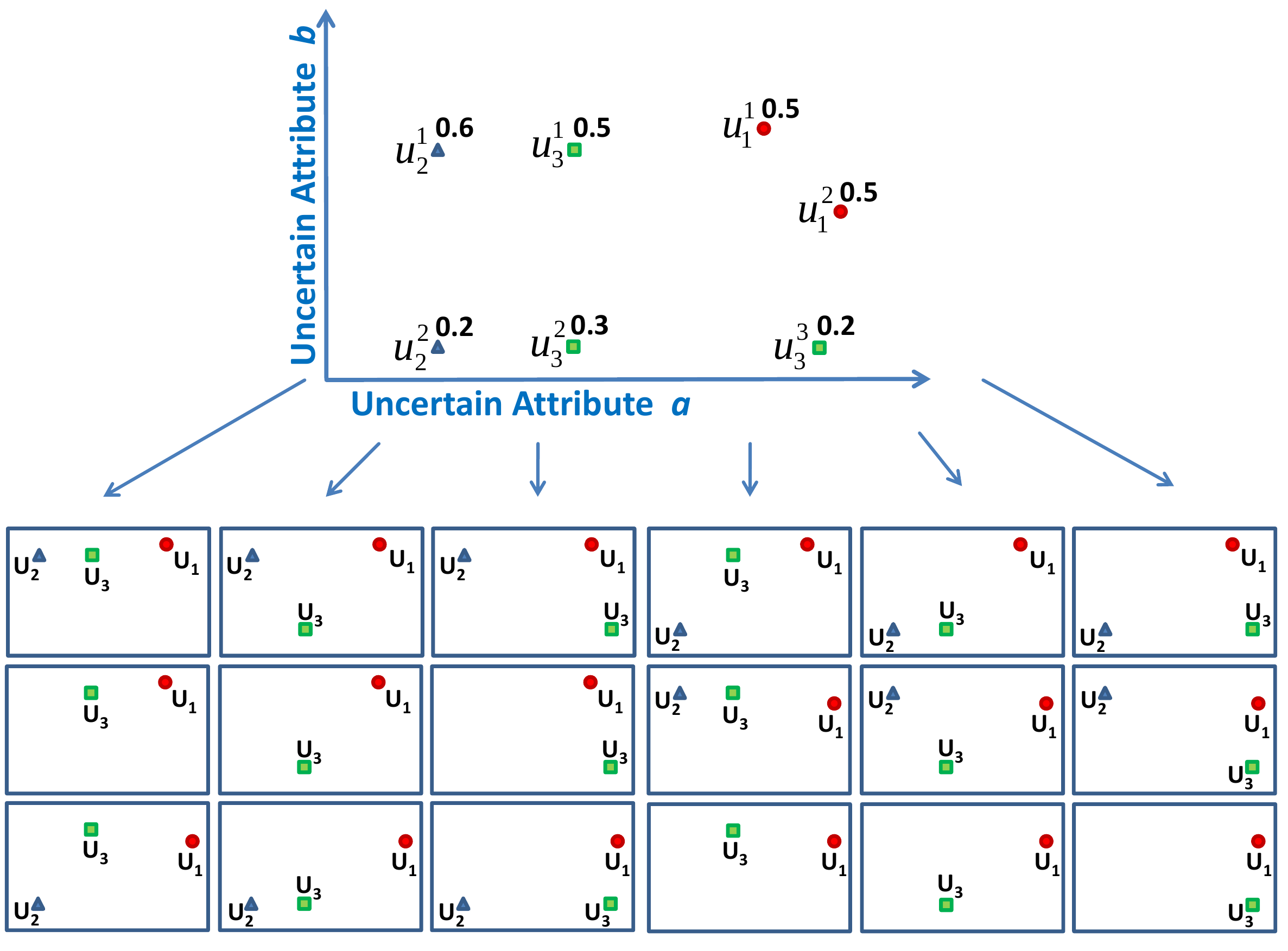}
    \caption{An uncertain database and all of its possible worlds.}
    \label{fig:pws_example}
\end{figure}

\begin{table}[htb]
\centering \caption{Possible worlds corresponding to Figure
\ref{fig:pws_example}.\vspace{-0.3cm}} \label{tab:pws_example}

\begin{tabular}{|c|c||c|c|}
\hline
\textbf{World} & \textbf{Probability} & \textbf{World} & \textbf{Probability}\\
\hline $\{u_1^1,u_2^1,u_3^1\}$ &  $0.5\cdot 0.7 \cdot 0.5=0.175$ &$\{u_1^2,u_2^1,u_3^1\}$ &  $0.5\cdot 0.7 \cdot 0.5=0.175$ \\
\hline $\{u_1^1,u_2^1,u_3^2\}$ &  $0.5\cdot 0.7 \cdot 0.3=0.105$ & $\{u_1^2,u_2^1,u_3^2\}$ &  $0.5\cdot 0.7 \cdot 0.3=0.105$ \\
\hline $\{u_1^1,u_2^1,u_3^3\}$ &  $0.5\cdot 0.7 \cdot 0.2=0.07$ & $\{u_1^2,u_2^1,u_3^3\}$ &  $0.5\cdot 0.7 \cdot 0.2=0.07$ \\
\hline $\{u_1^1,u_2^2,u_3^1\}$ &  $0.5\cdot 0.2 \cdot 0.5=0.05$ & $\{u_1^2,u_2^2,u_3^1\}$ &  $0.5\cdot 0.2 \cdot 0.5=0.05$ \\
\hline $\{u_1^1,u_2^2,u_3^2\}$ &  $0.5\cdot 0.2 \cdot 0.3=0.03$ & $\{u_1^2,u_2^2,u_3^2\}$ &  $0.5\cdot 0.2 \cdot 0.3=0.03$ \\
\hline $\{u_1^1,u_2^2,u_3^3\}$ &  $0.5\cdot 0.2 \cdot 0.2=0.02$ & $\{u_1^2,u_2^2,u_3^3\}$ &  $0.5\cdot 0.2 \cdot 0.2=0.02$ \\
\hline $\{u_1^1,u_3^1\}$ &  $0.5\cdot 0.1\cdot 0.5=0.025$ & $\{u_1^2,u_3^1\}$ &  $0.5\cdot 0.1 \cdot 0.5=0.025$ \\
\hline $\{u_1^1u_3^2\}$ &  $0.5\cdot 0.1 \cdot 0.3=0.015$ & $\{u_1^2,u_3^2\}$ &  $0.5\cdot 0.1 \cdot 0.3=0.015$ \\
\hline $\{u_1^1,u_3^3\}$ &  $0.5\cdot 0.1 \cdot 0.2=0.01$ & $\{u_1^2,u_3^3\}$ &  $0.5\cdot 0.1 \cdot 0.2=0.01$ \\

\hline

\end{tabular}
\vspace{-0.7cm}
\end{table}
\begin{example}
As an example, consider Figure \ref{fig:pws_example} where a
database consisting of three uncertain objects
$\DB=\{U_1,U_2,U_3\}$ is depicted. Objects $U_1=\{u_1^1,u_1^2\}$
and $U_2=\{u_2^1,u_2^2\}$ each have two possible instances, while
object $U_3=\{u_3^1,u_3^2,u_3^3\}$ has three possible instances.
The probabilities of these instances is given as
$P(u_1^1)=P(u_1^2)=0.5$, $P(u_2^1)=0.7$, $P(u_2^2)=0.2$,
$P(u_3^1)=0.5$, $P(u_3^2)=0.3$, $P(u_3^3)=0.2$. Note that object
$U_2$ is the only object having existential uncertainty: With a
probability of $1-0.7-0.2=0.1$ object $U_2$ does not exist at all.
Assuming independence between spatial objects, the probability for
the possible world where $U_1=u_1^1$, $U_2=u_2^1$ and $U_3=u_3^1$
is given by applying Equation \ref{eq:pw2} to obtain the product
$0.5\cdot 0.7 \cdot 0.5=0.175$. All possible worlds spanned by
$\DB$ are depicted in Figure \ref{fig:pws_example}. The
probability of each possible world is shown in Table
\ref{tab:pws_example}, including possible worlds where $U_2$ does
not exist.
\end{example}
Recall that a predicate can evaluate to either true or false on a
crisp (non-uncertain) database. An exemplary predicate is
\emph{There are at least five database objects in a 500meter range
of the location ``Theresienwiese, Munich''.} To evaluate a
predicate $\phi$ on an uncertain database using possible world
semantics, the query predicate is evaluated on each possible
world. The probability that the query predicate evaluates to true
is defined as the sum of probabilities of all worlds where $\phi$
is satisfied, formally:\vspace{-0.1cm}
\begin{definition} \label{def:querypws}
Let $\DB$ be an uncertain spatial database inducing the set of
possible worlds $\mathcal{W}$, let $\phi$ be some query predicate,
and let
$$\mathcal{I}(\phi,w\in\mathcal{W}):=P(\phi(\DB)|\DB=w)\in\{0,1\}$$ be the
indicator function that returns one if world $w$ satisfies $\phi$
and zero otherwise. The marginal probability $P(\phi(\DB))$ of the
event $\phi(\DB)$ that predicate $\phi$ holds in $\DB$ is defined as follows using the theorem of total probability \cite{Zwillinger00}:
\begin{equation}\label{eq:pwprob}
P(\phi(\DB))=\sum_{w\in\mathcal{W}}\mathcal{I}(\phi,w)\cdot P(w)\vspace{-0.1cm}
\end{equation}
\end{definition}
The main challenge of analyzing uncertain data is to efficiently
and effectively deal with the large number of possible worlds
induced by an uncertain database $\DB$. In the case of continuous
uncertain data, the number of possible worlds is uncountably
infinite and expensive integration operations or
numerical approximation are required for most spatial database
queries and spatial data mining tasks. Even in the case of
discrete uncertainty, the number of possible worlds grows
exponentially in the number of objects: in the worst case, any
combination of alternatives of objects may have a non-zero
probability, as shown exemplary in Figure \ref{fig:pws_example}.
This large number of possible worlds makes efficient query
processing and data mining an extremely challenging problem. In
particular, any problem that requires an enumeration of all
possible worlds is \#P-hard\footnote{\#P is the set of counting problems associated with decision problems in the class NP. Thus, for any NP-complete decision problem which asks if there exists a solution to a problem, the corresponding \#P problem asks for the number of such solutions. 
}.
In particular, a number of
probabilistic problems have been proven to be in \#P
\cite{Valiant79}. Following this argumentation, general query
processing in the case of discrete data using object independence
has proven to be a \#P-hard problem \cite{Dalvi04} in the context
of relational data. The spatial case is a specialization of the
relation case, but clearly, the spatial case is in \#P as well,
which becomes evident by construction of a query having an
exponentially large result, such as the query that returns all
possible worlds. Consequently, there can be no universal solution
that allows to answer \emph{any} query in polynomial time. This
implies that querying processing on models that are
generalizations of the discrete case with object independence,
e.g., models using continuous distribution, or models that relax
the object independence assumption, must also be a \#P hard
problem. The result of \cite{Dalvi04} implies that there exists
query predicates, for which no polynomial time solution can be
given. Yet, this result does not outrule the existence of query
predicates that can be answered efficiently. For example the
(trivial) query that always returns the empty set of objects can
be efficiently answered on uncertain spatial databases. 

\clearpage

\section{Existing Uncertain Spatial Database Management Systems}\label{sec:systems}
Recently developed systems provide support for spatio-temporal data in big data systems \cite{akdogan2010voronoi,aji2013hadoop,lu2012efficient,wang2010accelerating,zhang2012efficient}. Such systems exhibit high scalability for batch-processing jobs \cite{hadoop,dean2008mapreduce}, but do not provide efficient solutions to handle uncertain data and to assess the reliability of results.
The vivid field of managing, querying, and mining uncertain data has received tremendous attention from the database, data mining, and spatial data science communities. Recent books \cite{aggarwal2010managing} and survey papers  \cite{aggarwal2008survey,wang2013survey,Li2018Survey} provide an overview of the flurry of research papers that have appeared in these fields. 

been well-studied by the database research community in the past. 
While the traditional database
literature \cite{cavallo1987theory,barbara1992management,bacchus1996statistical,lakshmanan1997probview,fuhr1997probabilistic} has studied the problem of managing
uncertain data, this research field has seen a recent revival, due to modern techniques for
collecting inherently uncertain data. Most prominent concepts for probabilistic data management are MayBMS
\cite{antova2008fast}, MystiQ \cite{boulos2005mystiq}, Trio \cite{agrawal2006trio}, and BayesStore \cite{wang2008bayesstore}. These
uncertain database management systems (UDBMS) provide solutions to cope with uncertain
relational data, allowing to efficiently answer traditional queries that select subsets of data based on predicates or join different datasets based on conditions.
Extensions to the UDBMS also allow answering of important classes of spatial queries such as top-k and distance-ranking queries \cite{hua2008ranking,cormode2009semantics,li2009unified,bernecker2010scalable,li2010ranking}.
While these existing UDBMS provide probabilistic guarantees for their query results, they offer no support for data mining tasks. A likely reason for this gap
is the theoretic result of \cite{dalvi2007efficient} which shows that the problem of answering complex queries is
\#P-hard in the number of database objects. 
To illustrate this theoretic result, imagine running a simple range query with an arbitrary query point on a database having $N$ objects each having an arbitrary non-zero probability of being in that range. Further, assume stochastic independence between these objects. In that case, any of the $2^N$ combinations of result objects becomes a possible result and must be returned.

Nevertheless, a number of polynomial time solutions have been proposed in the literature for various spatial query types such as nearest neighbor queries \cite{ReyKalPra04,KriKunRen07,IjiIsh09,ChengCMC08}, k-nearest neighbor queries \cite{BesSolIly08,LjoSin07,LiSahDes09,CheCheCheXie09} and (similarity-) ranking queries \cite{BerKriRen08,CorLiYi09,LiSahDes09,SolIly09}. On first glance, these findings may look contradicting (unless $P=NP$), providing polynomial-time solution to a \#P-hard problem. On closer look, it shows that different related work use different semantics to interpret a result. Aforementioned related works that provide polynomial time solutions for spatial queries on uncertain data make a simplifying assumption: Rather than computing the probability for each possible result, they compute the probability of each \emph{object} to be part of the result. This reduces the number of probabilities that have to be reported, in the worst-case, from a number exponential in the number of database objects, to a linear number. Re-using the example of a range query on an uncertain database, it is possible to compute the probability that a single object is within the query range independent from all other objects.

Unfortunately, this simplification also yields a loss of information, as it is not possible to infer the probability of query results given only probabilities of individual objects. Let us revisit the running example from introduction, which is duplicate in Figure~\ref{fig:toy2} for convenience. This example will illustrate how such an object-based approach, which computes object-individual probabilities, rather than the probabilities of result sets, may yield misleading results.

\begin{figure}[t]
    \centering
    \includegraphics[width=0.75\textwidth]{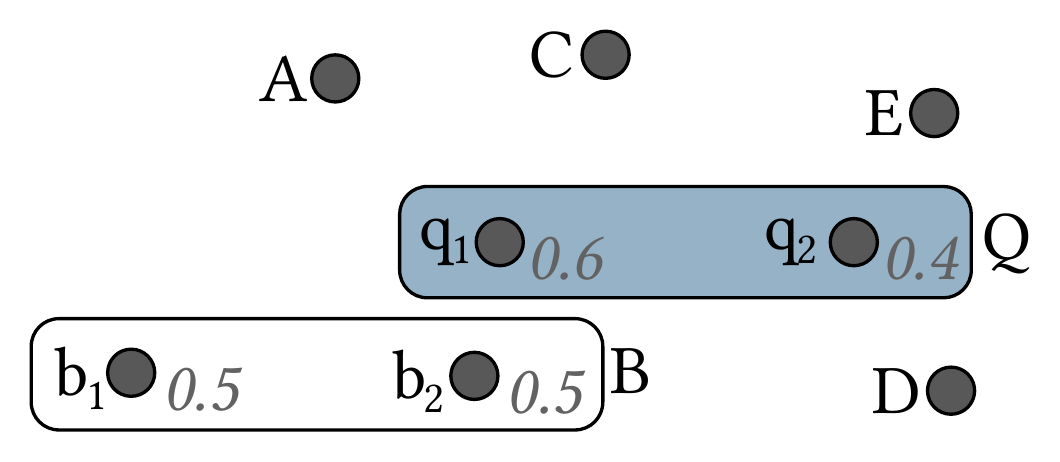}\vspace{-0.4cm}
    \caption{The Exemplary Uncertain Database from Figure~\ref{fig:toy}}\vspace{-0.5cm}
    \label{fig:toy2}
\end{figure}
\vspace{-0.1cm}
\begin{example}
Assume that the task is to simply find the probabilistic two nearest neighbors (2NN) of uncertain object $Q$. Objects $Q$ and $B$ have two alternative positions each, yielding a total of four possible worlds. For example, in one possible world, where $Q$ has location $q_1$ and $B$ has location $b_1$, the two nearest neighbors of $Q$ are $A$ and $C$. This possible world has a probability of $0.6\cdot 0.5=0.3$, obtained by assuming stochastic independence between objects. Following object-based result semantics, we can obtain probabilities of $0.3$, $0.3$, $0.6$, $0.4$, $0.4$ for objects $A$, $B$, $C$, $D$, and $E$ to be the $2$NNs of $Q$, respectively. However, this result hides any dependence between these result objects, such as objects $A$ and $B$ are mutually exclusive, while $D$ and $E$ are mutually inclusive.
\end{example}
\vspace{-0.1cm}

Towards approximate solutions, the Monte-Carlo DB (MCDB) system \cite{jampani2008mcdb} has been proposed, which samples possible worlds from the database, executes the query predicate on each sampled world.
MCDB estimates the probability of each object to be part of the result set. However, this approach of assigning a result probability to each object, as illustrate in the example above, cannot be extended to assess the probability of result sets. 
The problem is that the number of possible result sets may be exponentially large. To aggregate possible worlds into groups of mutually similar worlds (having similar results), an approach has been proposed for clustering of uncertain data~\cite{zufle2014representative,schubert2015framework} and more recently for general query processing on spatial data~\cite{Schmid2019Representative}. Revisiting the example of Figure~\ref{fig:toy}, this approach reports the results of a probabilistic query 2NN query as $\{A,C\}$, $\{B,C\}$, $\{D,E\}$, having respective probabilities of $0.3$, $0.3$, and $0.4$.
%
However, this approach (\cite{Schmid2019Representative}) can only be applied to spatial queries that return result sets, thus cannot be applied to more complex spatial queries or data mining tasks. 
To further elaborate the difference between solutions that compute the probability of each object to be part of the result, and solutions that compute the probability of each result, the following section will further survey the two different ``Probabilistic Result Semantics'': Object-based and Result-based.

\clearpage
\section{Probabilistic Result
Semantics}\label{subsec:ProbAnswerSem}


Recall that a spatial similarity query always requires a query object $q$ and, informally speaking, returns objects to the user that are similar to $q$. In the case of uncertain data, there exists two fundamental semantics to describe the result of such a probabilistic spatial similarity query. These different result semantics will
be denoted as \emph{object based result semantics} and the
\emph{result based result semantics}. Informally, the former semantics return possible \emph{result objects} and their probability of being part of the result, while the later semantics return possible results, which consist of a single object, of a set of objects or of a sorted list of objects depending on the query predicate, and their probability of being the result as a whole.

\subsection{Object Based Probabilistic Result Semantics} \label{sec:object_based}Using
\emph{object based probabilistic result semantics}, a
probabilistic spatial query returns a set of objects, each
associated with a probability describing the individual likelihood
of this object to satisfy the spatial query predicate.
\begin{definition}[Object Based Result Semantics]
Let $\DB$ be an uncertain spatial database, let $q$ be a query
object and let $\phi$ denote a spatial query predicate. Under
object based (OB) probabilistic result semantics, the result of a
probabilistic spatial $\phi$ query is a set
$\phi_{OB}(q,\DB)=\{(o\in\DB,P(o\in \phi_{OB}(q,\DB)))\}$ of
pairs. Each pair consists of a result object $o$ and its
probability $P(o\in\phi_{OB}(q,\DB))$ to satisfy $\phi$. Applying
possible world semantics (c.f. Definition \ref{ref:pws}) to
compute the probability $P(o\in\phi_{OB}(q,\DB))$ yields
\begin{equation}\label{eq:objectbased}
P(o\in\phi_{OB}(q,\DB))=\sum_{w\in\mathcal{W},o\in \phi(q,w)}P(w),
\end{equation}
where $\phi(q,w)$ is the deterministic result of a spatial $\phi$
query having query object $q$ applied to the deterministic
database defined by world $w$.
\end{definition}
Formally, the result of a probabilistic spatial query under object
based result semantics is a function
$$
\phi_{OB}(q,\DB) : \DB \rightarrow [0,1]
$$
$$
o\mapsto P(o\in\phi_{OB}(q,\DB)).
$$
mapping each object $o$ in $\DB$ (the results) to a probability value.
\clearpage

\begin{figure}[t]
    \centering
    \includegraphics[width=0.7\columnwidth]{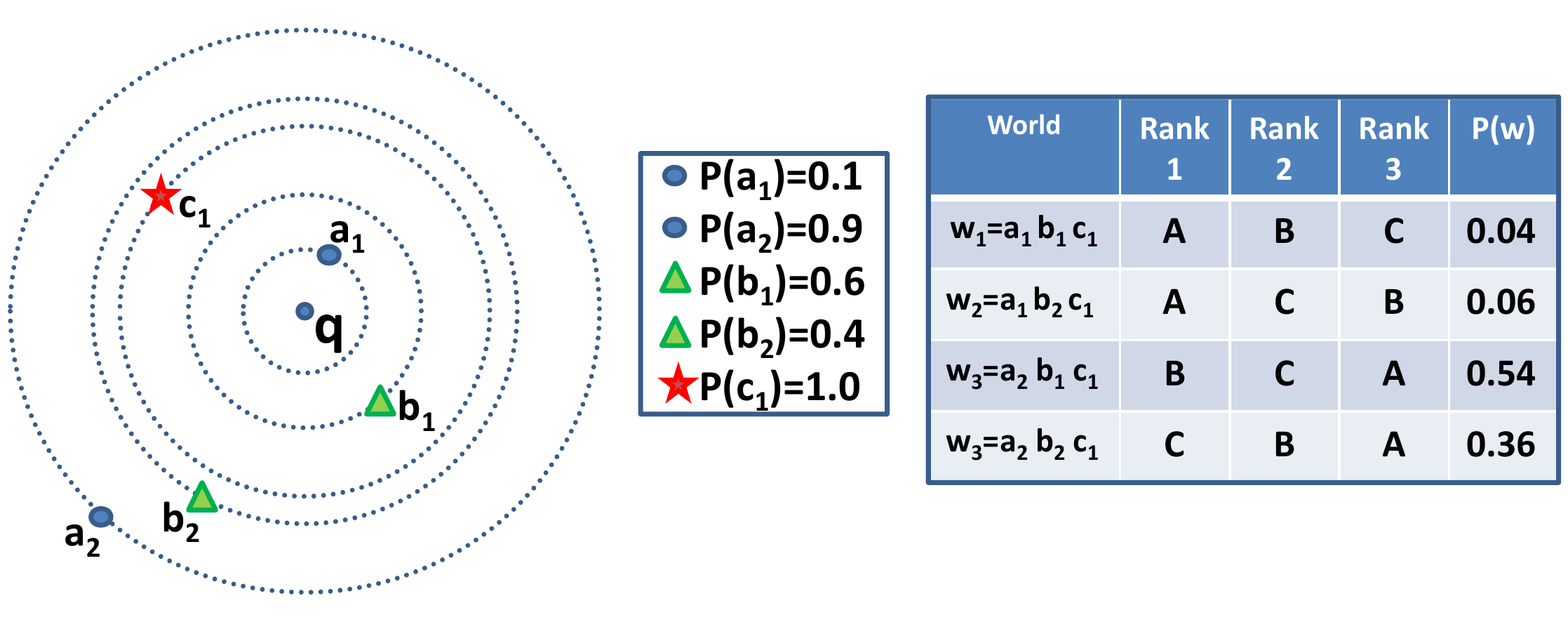}
    \caption{Example Database showing possible positions of uncertain objects and their corresponding probabilities.\vspace{-0.7cm}}
    \label{fig:ex_ranking}
\end{figure}

\begin{example}\label{ex:objectsemantics}
Figure \ref{fig:ex_ranking} depicts a database containing objects
$\DB=\{A,B,C\}$. Objects $A$ and $B$ have two alternative
locations each, while the position of $C$ is known for certain.
The locations and the probabilities of all alternatives are also
depicted in Figure \ref{fig:ex_ranking}. This leads to a total
number of four possible worlds. For example, in world $w_1$ where
$A=a_1$, $B=b_1$ and $C_1=c_1$, object $A$ is closest to $q$,
followed by objects $B$ and $C$. Assuming inter-object
independence, the probability of this world is given by the
product of individual instance probabilities $P(w_1)=P(a_1)\cdot
P(b_1)\cdot P(c_1)=0.04$. The ranking of each possible world and
the corresponding probability is also depicted in Figure
\ref{fig:ex_ranking}. For a probabilistic $2NN$ query for the
depicted query object $q$, the object based result semantic
computes the probability of each object to be in the two-nearest
neighbor set of $q$. For object $A$, the probability $P(A)$ of
this event equals $0.1$, since there exists exactly two possible
worlds $w_1$ and $w_2$ with a total probability of $0.04+0.06=0.1$
in which $A$ is on rank one or on rank two, yielding a result
tuple $(A,0.1)$. The complete result of a $P2NN$ query under
object based result semantics is $\{(A,0.1),(B,0.94),(C,0.96)\}$.
Note that in general, objects having a zero probability are
included in the result. For instance, assume an additional object
$D$ such that all instances of $D$ have a distance to $q$ greater
than the distance between $q$ and $b_2$. In this case, the pair
$(D,0)$ would be part of the result.
\end{example}
The result of a query under object based probabilistic result semantics contains one result tuple for every single
database object, even if the probability of the corresponding
object to be a result is very low or zero. In many applications,
such results may be meaningless. Therefore, the size of the result
set can be reduced by using a probabilistic query predicate as
explained later in Section \ref{sec:probpred}. A computational
problem is the computation of the probability $P(o\in\DB)$ of an
object $o$ to satisfy the spatial query predicate. In the example,
this probability was derived by iterating over the set of all
possible worlds $w_1,...,w_4$. Since this set grows exponentially
in the number of objects, such an approach is not viable in
practice. Therefore, efficient techniques to compute the
probability values $P(o)$ are required. A general paradigm to develop algorithms that avoid
an explicit enumeration of all possible worlds is presented in Section~\ref{part:paradigm}.

\subsection{Result Based Probabilistic Result Semantics} \label{sec:result_based}
 In the case of  result based result semantics, possible result
sets of a probabilistic spatial query are returned, each
associated with the probability of this result.
\begin{definition}[Result Based Result Semantics]\label{def:resultbased}
Let $\DB$ be an uncertain spatial database, let $q$ be a query
object and let $\phi$ denote a spatial query predicate. Under
result based (RB) result semantics, the result of a probabilistic
spatial $\phi$ query is a set
$$
\phi_{RB}(q,\DB)=\{(r,P(r))|r\subseteq\DB,
P(r)=\sum_{w\in\mathcal{W},\phi(q,w)=r}P(w)\}
$$
of pairs. This set contains one pair for each result
$r\subseteq\DB$ associated with the probability $P(r)$ of $r$ to
be the result. Following possible world semantics, the probability
$P(r)$ is defined as the sum of probabilities of all worlds
$w\in\mathcal{W}$ such that a spatial $\phi$ query returns $r$.
\end{definition}
Formally, the result of a probabilistic spatial query under result
based result semantics is a function
$$
\phi_{RB}(q,\DB) : \overline{\mathcal{P}}(\DB) \rightarrow [0,1]
$$
$$
r\mapsto P(r).
$$
mapping a elements of the power set $ \overline{\mathcal{P}}(\DB)$
(the results) to probability values.
\begin{example}\label{ex:resultsemantics}
For a probabilistic $2NN$ query for the depicted query object $q$,
result based result semantics require to compute the probability
of each subset of $\{A,B,C\}$ to be in the two-nearest neighbor
set of $q$. For the set $\{B,C\}$, the probability of this event
is $0.90$, since there is two possible worlds $w_3$ and $w_4$ with
a total probability of $0.54+0.36=0.9$ in which $B$ and $C$ are
both contained in the $2NN$ set of $q$. Note that in worlds $w_3$
and $w_4$ objects $B$ and $C$ appear in different ranking
positions. This fact is ignored by a $kNN$ query, as the results
are returned unsorted. In this example, the complete result of a
$P2NN$ query under object based result semantics is
$\{(\{A,B,C\},0)$, $(\{A,B\},0.04)$, $(\{A,C\},0.06)$,
$(\{B,C\},0.90)$, $(\{A\},0)$, $(\{B\},0)$, $(\{C\},0)$,
$(\{\emptyset\},0)\}$.
\end{example}
Clearly, the result of a query using result based result semantics
can be used to derive the result of an identical query using
object based result semantics. For instance, the result of Example
\ref{ex:resultsemantics} implies that the probability of object
$A$ to be a $2NN$ of $q$ is $0.10$, since there exists two
possible results using result based result semantics, namely
$(\{A,B\},0.04)$ and $(\{A,C\},0.06)$ having a total probability of
$0.04+0.06=0.1$, which matches the result of Example
\ref{ex:objectsemantics}.

\begin{lemma}\label{lem:OBRB}
Let $q$ be the query point of a probabilistic spatial $\phi$
query.
It holds that the result of this query using object based result
semantics $\phi_{OB}(q,\DB)$ is functionally dependent of the
result of this query using result based result semantics. The set
$PS\phi Q_{OB}(q,\DB)$ can be computed
 given only the set $PS\phi Q_{RB}(q,\DB)$ as follows:
$$
PS\phi Q_{OB}(q,\DB)=\{(o,P(o))|o\in\DB \wedge
P(o)=\sum_{(r,P(r))\in PS\phi Q_{RB}(q,\DB),o\in r}P(r)\}
$$
\end{lemma}
\begin{proof}
Let $\mathcal{W}$ denote the set of possible worlds of $\DB$, and
let $p(w\in\mathcal{W})$ denote the probability of a possible
world. Furthermore, let
$$
w_{S\subseteq\DB}:=\{w\in \mathcal{W}|\phi(q,w)=S\}
$$
denote the set of possible worlds such that $\phi(q,w)=S$, i.e.,
such that the predicate that a $\phi$ query using query object $q$
returns set $S$ holds. In each world $w$, query $q$ returns
exactly one deterministic result $PS\phi Q_{RB}(q,w)$. Thus, the
sets $w_{S\subseteq\DB}$ represent a complete and disjunctive
partition of $\mathcal{W}$, i.e., it holds that
\begin{equation}\label{eq:completepartitionworlds}\mathcal{W}=\bigcup_{S\subseteq\DB} w_S\end{equation} and \begin{equation}\label{eq:disjunctivepartitionworlds}\forall
R,S\in\overline{\mathcal{P}}(\DB): R\neq S\Rightarrow w_{R}\bigcap
w_{S}=\emptyset.\end{equation}

Using Equations \ref{eq:completepartitionworlds} and
\ref{eq:disjunctivepartitionworlds}, we can rewrite Equation
\ref{eq:objectbased}
$$P(o\in\phi_{OB}(q,\DB))=\sum_{w\in\mathcal{W},o\in
\phi(q,w)}P(w)$$ as
$$
P(o\in\phi_{OB}(q,\DB))=\sum_{S\in\overline{\mathcal{P}}(\DB)}\sum_{w\in
w_{S},o\in \phi(q,w)} P(w).
$$
By definition, query $q$ returns the same result for each world in
$w\in w_{S}$. This result contains object $o$ if $o\in S$. Thus we
can rewrite the above equation as
$$
P(o)=\sum_{S\in\overline{\mathcal{P}}(\DB),o\in S}P(S).
$$
The probabilities $P(S)$ are given by function $PS\phi
Q_{RB}(q,\DB)$.
\end{proof}
In the above proof, we have performed a linear-time reduction of
the problem of answering probabilistic spatial queries using
object based result semantics to the problem of answering
probabilistic spatial queries using result based result semantics.
Thus, we have shown that, except for a linear factor (which can be
neglected for most probabilistic spatial query types, since most
algorithm run in no better than log-linear time), the problem of
answering a probabilistic spatial query using result based result
semantics is at least as hard as answering a probabilistic spatial
query using object based semantics.

To summarize this section, we have learned about two different semantics to interpret the result of a spatial query on uncertain data: Object Based and Result Based. 
Understanding the difference of both result semantics is paramount to understand the landscape of existing research:
in some related publication the problem of answering some
probabilistic query may be proven to be in $\#P$, while another
publication gives a solution that lies in
$P$-TIME for the same spatial query predicate and the
same probabilistic query predicate. In such cases,
different result semantics may explain these results without
assuming $P=NP$.

\clearpage
\section{Probabilistic Query Predicates} \label{sec:probpred}
\begin{figure}[t]
    \centering
    \includegraphics[width=0.7\columnwidth]{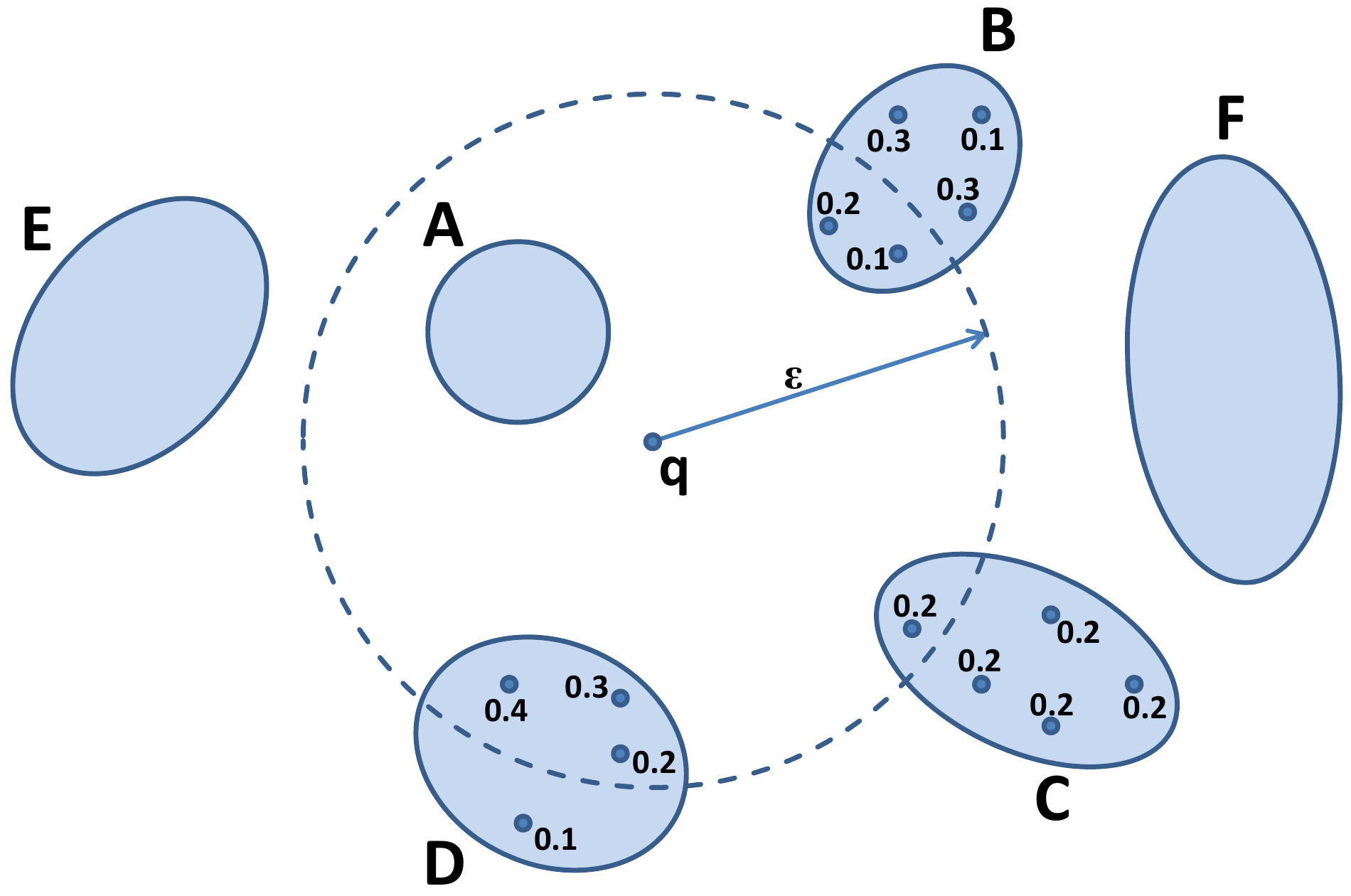}
    \caption{Example of an uncertain $\epsilon$-range query. Object $A$ is a true hit, objects $B$, $C$ and $D$ are possible hits.}
    \label{fig:blume}
\end{figure}
Generally, in an uncertain database, the question whether
an object satisfies a given query predicate $\phi$, such as being in a specified range or being a kNN of a query object, cannot
be answered deterministically due to uncertainty of
object locations. Due to this uncertainty, the predicate that an
object satisfies $\phi$ is a random variable, having some
(possibly zero, possibly one) probability. A probabilistic query
predicate quantifies the minimal probability required for a result
to qualify as a result that is sufficiently significant to be
returned to the user. This section formally define probabilistic query predicate for general query predicates.
The following definition are made for uncertain data in general, but can be applied analogously for uncertain spatial data.

A \emph{probabilistic query} can be defined
without any probabilistic query predicate. In this case, all
objects, and their respective probabilities are returned.
\begin{definition}[Probabilistic Query]
Let $\DB$ be an uncertain database, let $q$ be a query point and
let $\phi$ be a query predicate. A \emph{probabilistic query}
$\phi(q, \DB)$ returns all database objects $o\in\DB$ together
with their respective probability $P(o\in \phi(q,\DB))$ that $o$
satisfies $\phi$.
\begin{equation}\label{eq:Probabilistic Query}
\phi(q,\DB)=\{(o\in \DB,P(o\in \phi(q,\DB)))\}
\end{equation}
\end{definition}
The term $\emph{probabilistic}$ query is simply derived from the
fact that unlike a traditional query, a probabilistic query result
has probability values associated with each result. The main
challenge of answering a probabilistic query, is to compute the
probability $P(o\in \phi(q,\DB))$ for each object. Using possible
world semantics, a probabilistic query can be answered by
evaluating the query predicate for each object and each possible
world, i.e.,
$$
P(o\in \phi(q,\DB)):=\sum_{w_\in\mathcal{W}f_{ind}(\phi,w)\cdot
P(w)}.
$$
But clearly, it is necessary to avoid the combinatorial growth
that would be induced by this "naive" evaluation method.
\begin{example}
For example, consider the query \textit{``Return all friends of
user $q$ having a spatial distance of less than $100m$ to $q$''}
depicted in Figure \ref{fig:blume}. Thus, the predicate $\phi$ is
a $100$m-range predicate using query point $q$. We can
deterministically tell that friend $A$ must be within
$\epsilon=100m$ Euclidean distance of $q$, while friends $E$ and
$F$ cannot possibly be in range.
 The pairs $(A,1)$, $(E,0)$ and $(F,0)$ are added to the
result. For friends $B$, $C$ and $D$, this predicate cannot be
answered deterministically. Here, friend $B$ has some possible
positions located inside the $100m$ range of $q$, while other
possible positions are outside this range. The two locations
inside $q$'s range have a probability of $0.1$ and $0.2$,
respectively, thus the total probability of object $B$ to satisfy
the query predicate is $0.1+0.2=0.3$. The pair $(B,0.3)$ is thus
added to the result. The pairs $(C, 0.2)$ and $(D,0.9)$ complete
the result
$100\mbox{m-range}(q,\DB)=\{(A,1),(B,0.3),(C,0.2),(D,0.9),(E,0),(F,0)\}$.
\end{example}
The immediate question in the above example is:``Is a probability
of $0.3$ sufficient to warrant returning $B$ as a result?''. To
answer this question, a probabilistic query can explicitly specify
a probabilistic query predicate, to specify the requirements, in terms of probability, required for an object to qualify to be included in the result. The following subsections briefly survey the most commonly used probabilistic query predicates: probabilistic threshold queries and probabilistic Top$k$ queries. 

\subsection{Probabilistic Threshold Queries}\label{subsec:thresholdquery}
This paragraph defines a probabilistic query predicate that allows
to return only results that are statistically significant.
\begin{definition}[Probabilistic Threshold Query (P$\tau$Q)]
Let $\DB$ be an uncertain (spatial) database, let $q$ be a spatial
query object, let $0\leq \tau\leq 1$ be a real value and let
$\phi$ be a spatial query predicate. A \emph{probabilistic $\tau$
query} (P$\tau$Q) returns all objects $o\in\DB$ such that $o$ has
a probability of at least $\tau$ to satisfy $\phi(q,\DB)$:
$$
P\tau \phi(q,\DB):=\{o\in\DB|P(o\in\phi(q,\DB))\geq\tau\}.
$$
\end{definition}
\begin{example}
In Figure \ref{fig:blume}, a probabilistic threshold
$100m$-range(q,\DB) query with $\tau=0.5$ query returns the set of
objects $P0.5\,100\mbox{m-range}(q,\DB)=\{A,D\}$, since objects
$A$ and $D$ are the only objects such that their total probability
of alternatives inside the query region is equal or greater to
$\tau=0.5$.
\end{example}
Semantically, a probabilistic threshold spatial query returns all
results having a statistically significant probability to satisfy
the query predicate. Therefore, the probabilistic threshold query
serves as a statistical test of the hypothesis ``o is a result''
at a significance level of $\tau$. This test uses the probability
$P(o\in\phi(q,\DB))$ as a test statistic. Efficient algorithms to
compute this probability $P(o\in\phi(q,\DB))$, for the example of $k$NN and similarity ranking queries will be surveyed in Section~\ref{chap:sumofindependent} similarity ranking queries and R$k$NN
queries.

 A probabilistic threshold query on uncertain spatial
data is useful in applications, where the parameters of the
spatial predicate $\tau$ (e.g. the range of an $\epsilon$-range
query, or the parameter $k$ of a $k$NN query), as well as the
probabilistic threshold $\tau$ are chosen wisely, requiring expert
knowledge about the database $\DB$. If these parameters are chosen
inappropriately, no results may be returned, or the set of
returned result may grow too large. For example, if $\tau$ is
chosen very large, and if the database has a high grade of
uncertainty, then no result may be returned at all. Analogously, if the parameter
$\epsilon$ is chosen too small then no result will be returned,
while a too large value of $\epsilon$ may return all objects. The special case of having $\epsilon=0$, i.e., the case of returning all possible results (having a non-zero probability), is often used as default if no other probabilistic query predicate is specified (e.g. \cite{Soliman07,YiLKS08}). This case may be referred to as a \emph{possibilistic query predicate}, as all possible results (regardless of their probability) are returned.

\subsection{Probabilistic Top$k$ Queries}\label{subsec:topkquery}
In cases where insufficient information is given to select
appropriate parameter values, the following probabilistic query
predicate is defined to guarantee that only the $k$ most
significant results are returned.
\begin{definition}[Probabilistic Top$k$ Query (PTop$k$Q)]
Let $\DB$ be an uncertain spatial database, let $q$ be a spatial
query object, let $1\leq k \leq |\DB|$ be a positive integer, and
let $\phi$ be a spatial query predicate. A \emph{probabilistic
spatial Top$k$ query} (PTop$k$Q) returns the smallest set
PTop$k\phi$($q$,$\DB$) of at least $k$ objects such that
$$
\forall U_i\in \mbox{PTop}k\phi(q,\DB),
U_j\in\DB\setminus\mbox{PTop}k\phi(q,\DB):P(U_i\in
\phi(q,\DB))\geq P(U_j\in \phi(q,\DB))
$$
\end{definition}
Thus, a probabilistic spatial Top$k$ query returns the $k$ objects
having the highest probability to satisfy the query predicate.
Again, in case of ties, the resulting set may be greater than $k$.
\begin{example}
In Figure \ref{fig:blume}, a PTop$3\,\phi$ query using a
$\phi=100$m-range spatial predicate returns objects
$PTop3\,100\mbox{m-range}(q,\DB)=\{A,B,D\}$, since these objects
have the highest probability to satisfy the spatial predicate,
i.e., have the highest probability to be located in the spatial
$100m$-range.
\end{example}
Note, that the probabilistic Top$k$ query predicate can be
combined with a $kNN$ spatial query, i.e., with the case where
$\phi=kNN$. Such a probabilistic Top$k$ $jNN$ query returns the
set of $k$ objects having the highest probability, to be
$j$-nearest neighbor of the query object. Clearly, $k$ and $j$ may
have different integer values, such that differentiation is
needed.
\clearpage
\subsection{Discussion}
In summary, a probabilistic spatial query is defined by two query
predicates:
\begin{itemize}
\item A spatial predicate $\phi$ to select uncertain objects
having sufficiently \emph{high proximity} to the query object, and

\item a probabilistic predicate $\psi$, to select uncertain
objects having sufficiently \emph{high probability} to satisfy
$\phi$.
\end{itemize}
It has to be mentioned, that alternatively to this definition, a
single predicate can be used, that combines both spatial and
probabilistic features. For example, a monotonic score function
can be utilized, which combines spatial proximity and probability
to return a single scalar score. An example of such a monotone
score function is the expected distance function
$$
E(dist(q,U\in\DB))=\sum_{u\in U} P(u)\cdot dist(q,u),
$$
where $q$ is the query object, and $\DB$ is an uncertain database.
The expected support function is utilized by a number of related
publications, such as \cite{LjoSin07,CorLiYi09}. Using such a
monotone score function, objects with a sufficiently high score
can be returned. The advantage of using such an approach, is that
objects that are located very close to the query require a lower
probability to be returned as a result, while objects that are
located further away from the query object require a higher
probability. Yet, the main problem of such a combined predicate,
is that the probability of an object is treated as a simple
attribute, thus losing its probabilistic semantic. Thus, the
resulting score is very hard to interpret. An object that has a
high score, may indeed have a very low probability to exist at
all, because it is located (if it exists) very close to the query
object. Consequently, the score itself no longer contains any
confidence information, and thus, it is not possible to answer
queries according to possible world semantics using a single
aggregate, such as expected distance, only.

%% file: kap_03.tex
\section{The Paradigm of
Equivalent
Worlds}
\label{part:paradigm} In Section \ref{sec:pws} the concept of
possible world semantics has been introduced. Possible world
semantics give an intuitive and mathematically sound
interpretation of an uncertain spatial database. Furthermore,
queries that adhere to possible world semantics return unbiased
results, by evaluating the query on each possible world. Since any
such approach requires to run queries on an exponential number of
worlds, any naive approach is infeasible. Yet, for specific settings, such as specific result-based semantics, specific spatial query predicates and specific probabilistic query predicates, the literature has shown that it is possible to efficiently answer queries on uncertain data. 
While it is hardly feasible to enumerate all combinations of result-based semantics, spatial query predicates and probabilistic query predicates, this section introduces a general paradigm to find such a solution yourself. 
In a nutshell, the idea is to find, among the exponentially large set of possible worlds, a partitioning into a polynomially large number of subsets, which are equivalent for a given query. 

\subsection{Equivalent Worlds} The goal of this section is introduce a general paradigm to efficiently compute exact probabilities, while still
adhering to possible world semantics. For this purpose, reconsider Definition \ref{def:querypws}, defining the
probability that some predicate $\phi$ is satisfied in an
uncertain database $\DB$ as the total probability of all possible
worlds satisfying $\phi$. Recall Equation \ref{eq:pwprob}
$$
P(\phi(\DB))=\sum_{w\in\mathcal{W}}\mathcal{I}(\phi,w)\cdot P(w),
$$
where $\mathcal{W}$ is the set of all possible worlds;
$\mathcal{I}(\phi,w)$ is an indicator function that returns one if
predicate $\phi$ holds (i.e., resolves to true) in the crisp
database defined by world $w$ and zero otherwise, and $P(w)$ is
the probability of world $w$. To reduce the number of possible
worlds that need to be considered to compute $P(\phi(\DB))$, we
first need the following definition.

\begin{definition}[Class of Equivalent Worlds]\label{def:eqworlds}
Let $\phi$ be a query predicate and let $S\subseteq\mathcal{W}$ be
a set of possible worlds such that for any two worlds $w_1,w_2\in
S$ we can guarantee that $\phi$ holds in world $w_1$ if an only if
$\phi$ holds in world $w_2$, i.e.,
$$
\forall w_1,w_2\in S:\mathcal{I}(\phi,w_1) \Leftrightarrow
\mathcal{I}(\phi,w_2)
$$
Then set $S$ is called a \emph{class of worlds equivalent with
respect to} $\phi$. In the remainder of this chapter, if the spatial
query predicate $\phi$ is clearly given by the context, then $S$
will simply be denoted as a \emph{class of equivalent worlds}. Any
worlds $w_i,w_j\in S$ are denoted as \emph{equivalent worlds}.
\end{definition}
We now make the following observation:
\begin{corollary}\label{coro:eqworlds1}
Let $S\subseteq\mathcal{W}$ be a class of worlds equivalent with
respect to $\phi$ (c.f. Definition \ref{def:eqworlds}, we can
rewrite Equation \ref{eq:pwprob} as follows:
$$P(\phi(\DB))=\sum_{w\in\mathcal{W}}\mathcal{I}(\phi,w)\cdot
P(w) \Leftrightarrow$$
\begin{equation}\label{eq:pwprob2}
P(\phi(\DB))=\sum_{w\in\mathcal{W}\setminus
S}\mathcal{I}(\phi,w)\cdot P(w)+\mathcal{I}(\phi,w\in S)\cdot
\sum_{w\in S}P(w). \end{equation}
\end{corollary}

\clearpage
\begin{proof}
Due to the assumption that for any two worlds $w_1,w_2\in S$ it
holds that $\phi$ holds in world $w_1$ if an only if $\phi$ holds
in world $w_2$, we get $\mathcal{I}(\phi,w_1)=1 \Leftrightarrow
\mathcal{I}(\phi,w_2)=1$ and $\mathcal{I}(\phi,w_1)=0
\Leftrightarrow \mathcal{I}(\phi,w_2)=0$ by definition of function
$\mathcal{I}$. Due to this assumption, we have to consider two
cases.

{\bf Case 1:} $\forall w\in S:\mathcal{I}(\phi,w)=0$

In this case, both Equation \ref{eq:pwprob} and Equation
\ref{eq:pwprob2} can be rewritten as
$$
P(\phi(\DB))=\sum_{w\in\mathcal{W}\setminus
S}\mathcal{I}(\phi,w)\cdot P(w).
$$

{\bf Case 2:} $\forall w\in S:\mathcal{I}(\phi,w)=1$

In this case, both Equation \ref{eq:pwprob} and Equation
\ref{eq:pwprob2} can be rewritten as
$$
P(\phi(\DB))=\sum_{w\in\mathcal{W}\setminus
S}\mathcal{I}(\phi,w)\cdot P(w)+\sum_{w\in S}P(w)
$$\qed
\end{proof}
The only difference between both cases is the additive term
$\sum_{w\in S}P(w)$, which exists only in Case 2. The indicator
function $\mathcal{I}(\phi,w\in S)$ ensures that this term is only
added in the second case. As main purpose, Corollary
\ref{coro:eqworlds1} states that, given a set of equivalent worlds
$S$, we only have to evaluate the indictor function
$\mathcal{I}(\phi,w)$ on a single representative world $w\in S$,
rather than on each world in $S$. This allows to reduce the number
of (crisp) $\phi$ queries required to compute Equation
\ref{eq:pwprob} by $|S|-1$.

Corollary \ref{coro:eqworlds1} leads to the following Lemma.
\begin{lemma}\label{lem:pwpartition}
Let $\mathcal{S}$ be a partitioning of $\mathcal{W}$ into
disjoint sets such that $\bigcup_{S\in
\mathcal{S}}S=\mathcal{W}$ and for all
$S_1,S_2\in\mathcal{S}:S_1\cap S_2=\emptyset$. Equation
\ref{eq:pwprob} can be rewritten as
$$P(\phi(\DB))=\sum_{w\in\mathcal{W}}\mathcal{I}(\phi,w)\cdot
P(w) \Leftrightarrow$$
\begin{equation}\label{eq:pwprobreformed}
P(\phi(\DB))=\sum_{S\in\mathcal{S}}\mathcal{I}(\phi,w\in S)\cdot
\sum_{w\in S}P(w).
\end{equation}
\end{lemma}
\begin{proof}
Lemma \ref{lem:pwpartition} is derived by applying Corollary
\ref{coro:eqworlds1} once for each $S\in\mathcal{S}$.\qed
\end{proof}
The next subsection will show how to leverage Lemma~\ref{lem:pwpartition} to partition the set of all possible worlds into equivalence classes that are guaranteed to have the same result for a given query predicate, and how to exploit this partitioning to efficiently answer queries.
\clearpage
\subsection{Exploiting Equivalent Worlds for Efficient Algorithms}
 Given a partitioning $\mathcal{S}$ of all possible worlds, Equation \ref{eq:pwprobreformed} requires to perform
 the following two tasks. The first task requires to evaluate the
 indicator function $\mathcal{I}(\phi,w\in S)$ for one representative
 world of each partition. This can be achieved by performing a
 traditional (non-uncertain) $\phi$ query on these representatives.
The final challenge is to efficiently compute the total
probability $P(S):=\sum_{w\in S}P(w)$ for each equivalent class
$S\in\mathcal{S}$. This computation must avoid an enumeration of
all possible worlds, i.e., must be in $o(|S|)$.\footnote{Note that
if an exponential large set is partitioned into a polynomial
number of subsets, then at least one such subset must have
exponential size. This is evident considering that
$O(\frac{2^n}{poly(n)})=O(2^n)$.} Achieving an efficient
computation is a creative task, and usually requires to exploit
properties of the model (such as object independence) and
properties of the spatial query predicate.
\begin{figure}[t]
    \centering
    \includegraphics[width=0.8\columnwidth]{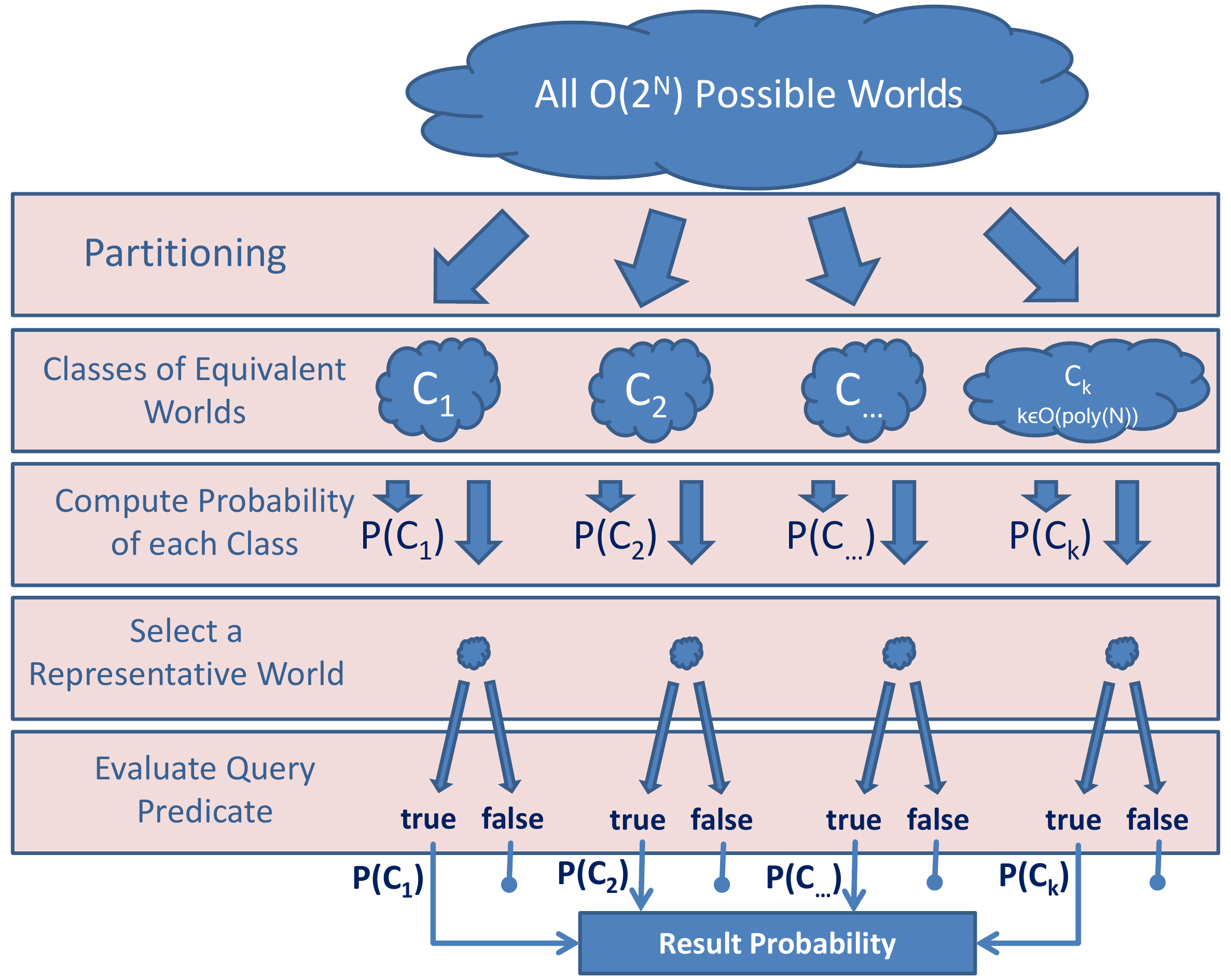}
    \caption{Summary of the Paradigm of Equivalent Worlds.}
    \label{fig:paradigm}
\end{figure}
The paradigm of equivalent worlds is illustrated and summarized in
Figure \ref{fig:paradigm}. In the first step, set of all possible
worlds $\mathcal{W}$, which is exponential in the number $N$ of
uncertain objects, has to be partitioned into a polynomial large
set of classes of equivalent worlds, such that all worlds in the
same class are guaranteed to be equivalent given the query
predicate $\phi$. This yields a the set
$\mathcal{C}=\{C_1,C_2,...,C_k\}$ of classes of equivalent worlds.
To allow efficient processing, this set must be polynomial in
size, since each class has to be considered individually in the
following. Next, we require to compute the probability of each
class $C_i$, without enumeration of all possible worlds contained
in $C_i$, the number of which may still be exponential. In fact,
at least one class $C_i$ must contain $O(2^N)$ possible worlds.
Next, we need to decide, for each class $C_i$, whether all worlds
$w\in C_i$ satisfy the query predicate $\phi$, or whether no world
$w\in C_i$ satisfies $\phi$. Due to equivalence of all possible
worlds in $C_i$, these are the only possible cases. For some query
predicates, this decision can be made using special properties of
the query predicate, as we will see later in this chapter. In the
general case, this decision can be made by choosing one
representative world $w\in C_i$ (e.g. at random) from each class
$C_i$, and evaluating the query predicate on this world. This
yields at total run-time of $O(|\mathcal{C}|)\cdot
O(\mathcal{I}(\phi,w))$, where $\mathcal{I}(\phi,w)$ is the time
complexity of evaluating the query predicate $\phi$ on the certain
database $w$. If this query predicate can be evaluated in
polynomial time, i.e., if $O(\mathcal{I}(\phi,w))\in O(poly(N))$,
then the total run-time is in $O(poly(N))$. This is
evident, since if $O(\mathcal{C})$ is in $O(poly(N))$, then
$O(\mathcal{C})\cdot O(\mathcal{I}(\phi,w))$ is in
$O(poly(N))\cdot O(poly(N))$ which is in $O(poly(N))$. For each
class $C_i$, where the representative world satisfies $\phi$, the
corresponding probability $P(C_i)$ is added to the result
probability.

The following lemma summarizes the assumptions that a query
predicate has to satisfy in order to efficiently apply paradigm of
finding equivalent worlds.
\begin{lemma}\label{lem:paradigm}
Given a query predicate $\phi$ and an uncertain database $\DB$ of size $N:=|DB|$, we
can answer $\phi$ on $\DB$ in polynomial time if the following four
conditions are satisfied:
\begin{itemize}
\item[I] A traditional $\psi$ query on non-uncertain data can be
answered in polynomial time.

 \item[II] we can identify a
partitioning $\mathcal{C}$ of $\mathcal{W}$ into classes
$C\in\mathcal{C}$ of equivalent worlds (see Definition
\ref{def:eqworlds}.

\item[III] The number $|\mathcal{C}|$ of classes is at most
polynomial in $N$.

\item[IV] The the total probability of a class $S\in\mathcal{C}$
can be computed in at most polynomial time.
\end{itemize}
\end{lemma}
\begin{proof}
Answering a $\phi$ query on $\DB$ requires to evaluate Equation
\ref{eq:pwprob} which we reformed into Equation
\ref{eq:pwprobreformed} using Property II. This requires to
iterate over all $|\mathcal{C}|$ classes of equivalent worlds in polynomial
time due to Property III. For each class $C\in\mathcal{C}$, this requires to
perform two tasks. The first task requires to compute the total
probability of all worlds in $C$, and the second task requires to
evaluate $\phi$ on a single possible world $w\in C$. The former
task can be performed in polynomial time due to Property IV. The
later task requires to perform a crisp $\phi$ query on the (crisp)
world $w$ in polynomial time due to Property I.
\end{proof}

\clearpage
\section{Case Study: Range Queries and the Sum of Independent Bernoulli Trials}\label{chap:sumofindependent}
In this chapter, the paradigm of equivalent worlds will be applied
to efficiently solve the problem of computing the number of uncertain objects located within a specified range.

\begin{example}
As an example, consider the setting depicted in Figure~\ref{fig:blume}. In this example, we have four objects, $A$, $B$, $C$, and $D$ having probabilities of $1.0$, $0.3$, $0.2$, and $0.9$ of being located inside the query region defined by query location $q$ and query range $\epsilon$.
Intuitively, the number of objects in this range can be anywhere between one and four, as only object $A$ is guaranteed to be inside the range, while on $B$, $C$, and $D$ have a chance to be inside this range among all other objects. How can we efficiently compute the distribution of this number of objects inside the query range? What is the probability of having exactly one, two, three or four object in the range? Intuitively, the number of objects corresponds depends on the result of three ``coin-flips'', each using a coin with a different bias of flipping heads. 
\end{example}
Each such ``coin-flip'' is a Bernoulli trial, which may have a successful (``heads'') of unsuccessful (``tails'') outcome. In the case where all Bernoulli trials have the same probability $p$, the number of successful trials out of $N$ trials is described by the well-known binomial distribution.
In the case where each trial may have a different probability to succeed, the number of successful trials follows a Poisson-binomial distributions \cite{hoeffding1956distribution}.

Formally, let $X_1,...,X_N$ be independent and not necessarily identically
distributed Bernoulli trials, i.e., random variables that may only
take values zero and one. Let $p_i:=P(X_i=1)$ denote the
probability that random variable $X_i$ has value one. In this
section, we will show how to efficiently compute the distribution
of the random variable
$$
\sum_{i=1}^N X_i
$$
without enumeration of all possible worlds. That is, for each
$0\leq k\leq N$, this section shows how to compute the probability
$P(\sum_{i=1}^N X_i=k)$ that exactly $k$ trials are successful.

This section shows two commonly used solutions to compute the distribution of $\sum_i X_i$ efficiently: The Poisson-binomial recurrence, and a technique based on generating functions. Both solutions have in common that they 
identify worlds that are equivalent to the query predicate.
\clearpage

%
%
\begin{figure}[t]
    \centering
    \includegraphics[width=0.7\columnwidth]{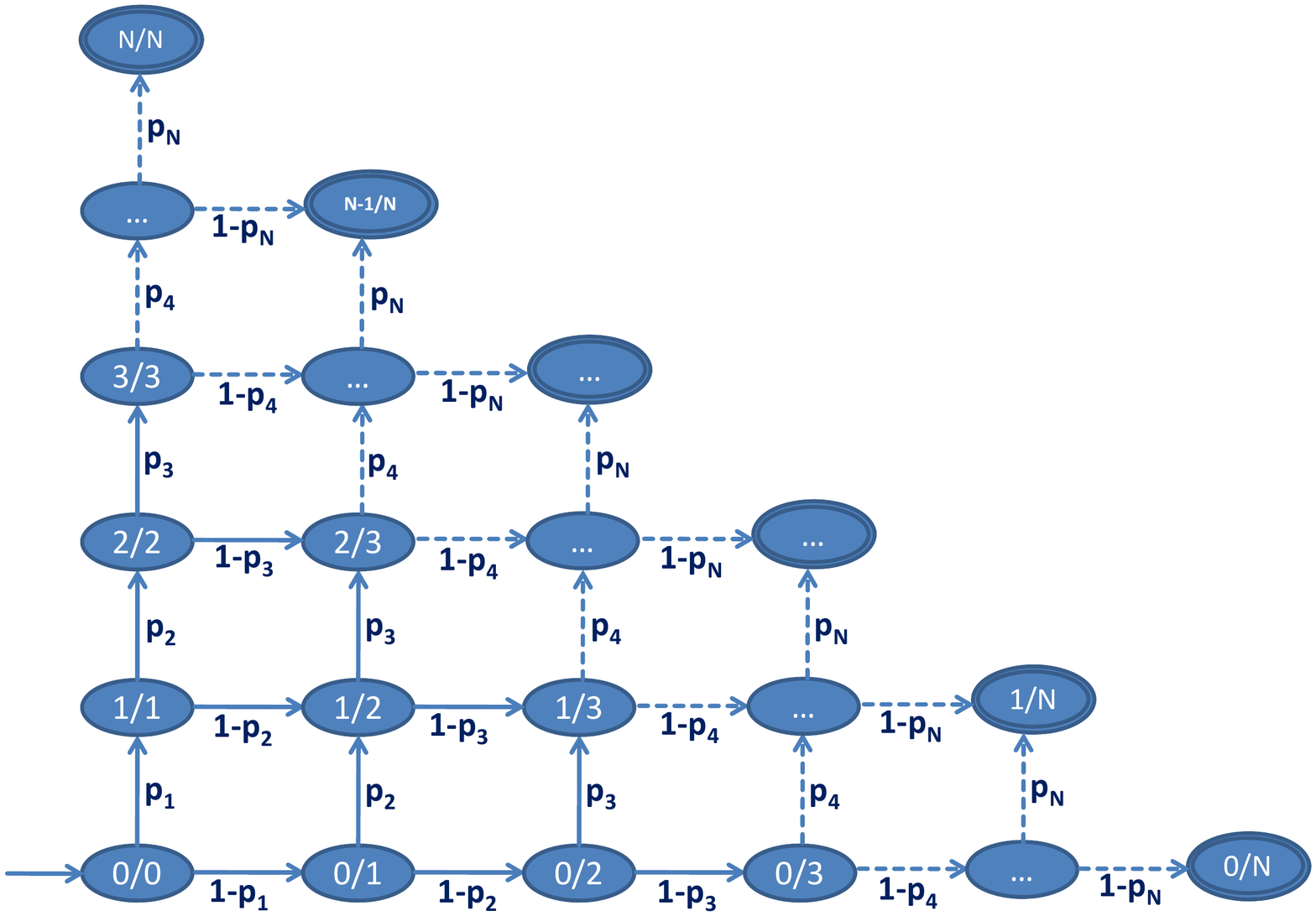}
    \caption{Deterministic finite automaton corresponding to the problem of the sum of independent Bernoulli trials.}
    \label{fig:PBR_automaton}
\end{figure}

\subsection{Poisson-Binomial Recurrence}\label{paradigm:pbr}
The first approach iteratively computes the distribution of the
 sum of the first $1\leq k \leq N$ Bernoulli variables given
the distribution of the sum of the first $k-1$ Bernoulli
variables.

 To gain an intuition of how to do this
efficiently, consider the deterministic finite automaton depicted
in Figure \ref{fig:PBR_automaton}.\footnote{Note that this
automaton is deterministic, despite the process of choosing a
successor node being a random event. Once the Bernoulli trial
corresponding to a node has been performed, the next node will be
chosen deterministically, i.e., the upper node will be chosen if
the trial was successful, and the right node will be chosen
otherwise. Either way, there is exactly one successor node.} The
states (i/j) of this automaton correspond to the random event that
out of the first $j$ Bernoulli trials $X_1,...,X_j$, exactly $i$
trials have been successful. Initially, zero Bernoulli trials have
been performed, out of which zero (trivially) were successful.
This situation is represented by the initial state $(0/0)$ in
Figure \ref{fig:PBR_automaton}. Evaluating the first Bernoulli
trial $X_1$, there is two possible outcomes: The trial may be
successful with a probability of $p_1$, leading to a state $(1/1)$
where one out of one trials have been successful. Alternatively,
the trial may be unsuccessful, with a probability of $1-p_1$,
leading to a state $(0/1)$ where zero out of one trial have been
successful. The second trial is then applied to both possible
outcomes. If the first trial has not been successful, i.e., we are
currently located in state $(0/1)$, then there is again two
outcomes for the second Bernoulli trial, leading to state $(1/2)$
and $(0/2)$ with a probability of $p_2$ and $1-p_2$ respectively.
If currently located in state $(0/1)$, the two outcomes are state
$(2/2)$ and state $(1/2)$ with the same probabilities. At this
point, we have unified two different possible worlds that are
equivalent with respect to $\sum_i X_i$: The world where trial one
has been successful and trial two has not been successful, and the
world where trial one has not been successful and trial two has
been successful have been unified into state $(1/2)$, representing
both worlds. This unification was possible, since both paths
leading to state $(1/2)$ are equivalent with respect to the number
of successful trials.

The three states $(0/2)$, $(1/2)$ and $(2/2)$ are then subjected
to the outcome of the third Bernoulli trial, leading to states
$(0/3)$, $(1/3)$, $(2/3)$ and $(3/3)$. That is a total of four
states for a total of $2^3=8$ possible worlds. In summary, the
number of states in Figure \ref{fig:PBR_automaton} equals
$\frac{N^2}{2}$. However, it is not yet clear how to compute the
probability of a state $(i/j)$ efficiently. Naively, we have to
compute the sum over all paths leading to state $(i/j)$. For
example, the probability of state (2/3) is given by $p_1\cdot p_2
\cdot (1-p_3)+p_1\cdot (1-p_2)\cdot p_3+(1-p_1)\cdot p_2\cdot
p_3$. This naive computation requires to enumerate all $\binom{j}{p_3}$ combinations of paths leading to state $(i/j)$.

For an efficient computation, we make the following observation:
Each state of the deterministic finite automaton depicted in Figure \ref{fig:PBR_automaton}
has at most two incoming edges. Thus, to compute the probability
of a state $(i/j)$, we only require the probabilities of states
leading to $(i/j)$. The states leading to state $(i/j)$ are state
$(i-1/j-1)$ and state $(i/j-1)$. Given the probabilities
$P(i-1/j-1)$ and $P(i/j-1)$, we can compute the probability
$P(i/j)$ of state (i/j) as follows:
\begin{equation}\label{eq:pbr}
P(i/j)=P(i-1/j-1)\cdot p_j + P(i,j-1)\cdot (1-p_j)
\end{equation}
where
$$
P(0/0)=1 \mbox{ and } P(i/j)=0\mbox{ if }i>j\mbox{ or if }i<0.
$$
Equation \ref{eq:pbr} is known as the Poisson-Binomial Recurrence
(To the best of our knowledge, the Poisson binomial recurrence was
first introduced by \cite{Lange99}) and can be used to compute the
probabilities of states $(k/N),0\leq k\leq N$ which by definition,
correspond to the probabilities $P(\sum_{i=1^N}X_i=k)$ that out of
all $N$ Bernoulli trials, exactly $k$ trials are successful.

This approach follows the paradigm of equivalent worlds in each
iteration $k$: The set of all $2^k$ possible worlds is partitioned
into $k+1$ equivalent sets, each corresponding to a state $i/k$,
where $i\leq k$. Each class contains only and all of the $\binom{k}{i}$ possible worlds where exactly $i$ Bernoulli trails
succeeded. The information about the particular sequence of the
successful trials, i.e., which trials were successful and which
were unsuccessful is lost. This information however, is no longer
necessary to compute the distribution of $\sum_{i=0}^{N}X_i$,
since for this random variable, we only need to know the number of
successful trials, not their sequence. This abstraction allows to
remove the combinatorial aspect of the problem.

An example showcasing the Poisson binomial recurrence is given in
the following.

\begin{figure}[t]
    \centering
    \includegraphics[width=0.7\columnwidth]{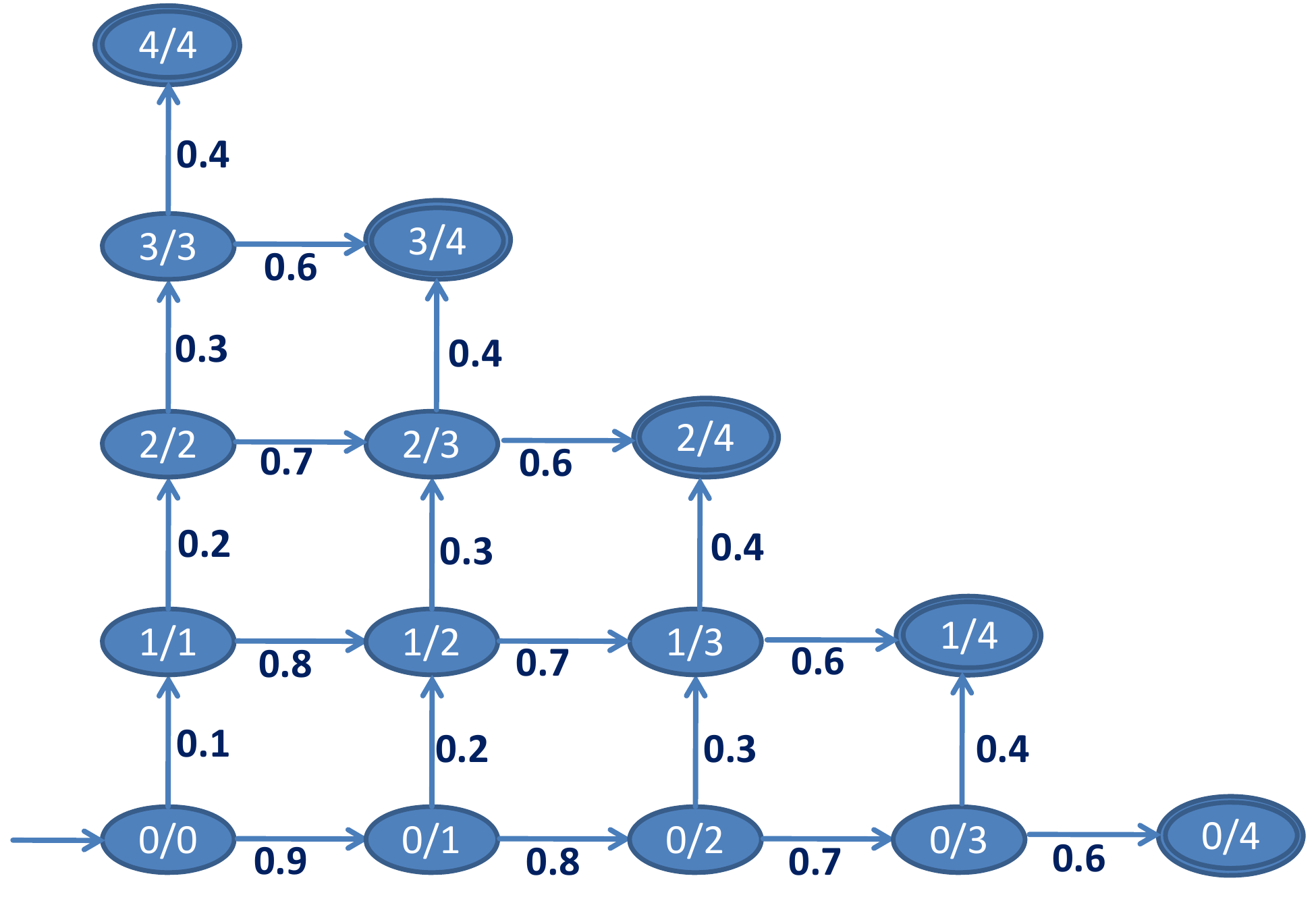}\vspace{-0.2cm}
    \caption{Deterministic finite automaton for four Bernoulli random variables.}\vspace{-0.2cm}
    \label{fig:PBR_automaton_example}
\end{figure}

\begin{example}\label{ex:pbr123}
Let $N=4$ and let $p_1=0.1$, $p_2=0.2$, $p_3=0.3$ and $p_4=0.4$.
The corresponding DFA is depicted in Figure
\ref{fig:PBR_automaton_example}. The probability of state (0/0) is
explicitly set to $1.0$ in Equation \ref{eq:pbr}. To compute the
probability of state (0/1), we apply Equation \ref{eq:pbr} to
compute
$$
P(0/1)= P(-1/0)\cdot p_1+P(0/0)\cdot (1-p_1).
$$
with $P(-1/0)=0$ and $P(0/0)=1$ explicitly defined in Equation
\ref{eq:pbr} this yields
$$
P(0/1)=0\cdot p_1 + 1\cdot (1-p_1)=0.9
$$
Analogously, we obtain
$$
P(1/1)=P(0/0)\cdot p_1 + P(1/0)\cdot (1-p_1)=1\cdot p_1=0.1
$$
Using these initial probabilities, we can continue to compute
$$
P(0/2)=P(-1/1)\cdot p_2 + P(0/1)\cdot (1-p_2)=0\cdot 0.2 +0.9\cdot
0.8=0.72
$$
$$
P(1/2)=P(0/1)\cdot p_2 + P(1/1)\cdot (1-p_2)=0.9\cdot 0.2+0.1\cdot
0.8= 0.26
$$
$$
P(2/2)=P(1/1)\cdot p_2 + P(2/1)\cdot (1-p_2)=0.1\cdot 0.2 + 0\cdot
0.8= 0.02
$$
The probabilities $P(i/2),0\leq i\leq 2$ can be used to compute
$$
P(0/3)=P(-1/2)\cdot p_3+P(0/2)\cdot (1-p_3)=0\cdot 0.3+0.72\cdot
0.7=0.504
$$
$$
P(1/3)=P(0/2)\cdot p_3+P(1/2)\cdot (1-p_3)=0.72\cdot 0.3+0.26\cdot
0.7=0.398
$$
$$
P(2/3)=P(1/2)\cdot p_3+P(2/2)\cdot (1-p_3)=0.26\cdot 0.3+0.02\cdot
0.7=0.092
$$
$$
P(3/3)=P(2/2)\cdot p_3+P(3/2)\cdot (1-p_3)=0.02\cdot 0.3+0\cdot
0.7=0.006
$$
Finally, these probabilities can be used to derive the final
distribution of the random variable $\sum_{i=1}^4 X_i$:
$$
P(0/4)=P(-1/3)\cdot p_4+P(0/3)\cdot (1-p_4)=0\cdot 0.4+0.504\cdot
0.6=0.3024
$$
$$
P(1/4)=P(0/3)\cdot p_4+P(1/3)\cdot (1-p_4)=0.504\cdot
0.4+0.398\cdot 0.6=0.4404
$$
$$
P(2/4)=P(1/3)\cdot p_4+P(2/3)\cdot (1-p_4)=0.398\cdot
0.4+0.092\cdot 0.6=0.2144
$$
$$
P(3/4)=P(2/3)\cdot p_4+P(3/3)\cdot (1-p_4)=0.092\cdot
0.4+0.006\cdot 0.6=0.0404
$$
$$
P(4/4)=P(3/3)\cdot p_4+P(4/3)\cdot (1-p_4)=0.006\cdot 0.4+0\cdot
0.6=0.0024
$$
These probabilities describe the PDF of $\sum_{i=1}^4 X_i$ by definition of $P(i/j)$.
\end{example}

\subsubsection*{Complexity Analysis} To compute the distribution of
$\sum_i X_i$ we require to compute each probability $P(i/j)$ for
$0\leq j\leq N,i\leq j$, yielding a total of $\frac{N^2}{2}\in
O(N^2)$ probability computations. To compute any such probability,
we have to evaluate Equation \ref{eq:pbr}, which requires to look
up four probabilities $P(i-1/j-1)$, $P(i/j-1)$, $p_j$ and $1-p_j$,
which can be performed in constant time. This yields a total
runtime complexity of $O(N^2)$. The $O(N^2)$ space complexity
required to store the matrix of probabilities $P(i/j)$ for $0\leq
j\leq N,i\leq j$ can be reduced to $O(N\cdot k)$ by exploiting that in
each iteration where the probabilities $P(i/k),0\leq i\leq k$ are
computed, only the probabilities $P(i/k-1),0\leq i\leq k-1$ are
required, and the result of previous iterations can be discarded.
Thus, at most $N$ probabilities have to be stored at a time.

\vspace{-0.3cm}
\subsection{Generating Functions}\vspace{-0.2cm} \label{paradigm:gfs} An
alternative technique to compute the sum of independent Bernoulli
variables is the generating functions technique. While showing the
same complexity as the Poisson binomial recurrence, its advantage
is its intuitiveness.

Represent each Bernoulli trial $X_i$ by a polynomial
$poly(X_i)=p_i\cdot x+(1-p_i)$. Consider the generating function
\begin{equation}\label{eq:genfkts}
\mathcal{F}^N=\prod_{i=1}^N poly(X_i)=\sum_{i=0}^{N} c_ix^i.
\end{equation}
The coefficient $c_i$ of $x^i$ in the expansion of $\mathcal{F}^N$
equals the probability $P(\sum_{n=1}^N X_n=i)$ (\cite{LiD09}). For
example, the monomial $0.25\cdot x^4$ implies that with a
probability of $0.25$, the sum of all Bernoulli random variables
equals four.

The expansion of $N$ polynomials, each containing two monomials
leads to a total of $2^N$ monomials, one monomial for each
sequence of successful and unsuccessful Bernoulli trials, i.e.,
one monomial for each possible worlds. To reduce this complexity,
again an iterative computation of $\mathcal{F}^N$, can be used, by
exploiting that
\begin{equation}\label{eq:genfkts2}
\mathcal{F}^k=\mathcal{F}^{k-1}\cdot poly(X_k).
\end{equation}
This rewriting of Equation \ref{eq:genfkts} allows to inductively
compute $\mathcal{F}^k$ from $\mathcal{F}^{k-1}$. The induction is
started by computing the polynomial $\mathcal{F}^0$, which is the
empty product which equals the neutral element of multiplication,
i.e., $\mathcal{F}^0=1$. To understand the semantics of this
polynomial, the polynomial $\mathcal{F}^0=1$ can be rewritten as
$\mathcal{F}^0=1\cdot x^0$, which we can interpret as the
following tautology:``with a probability of one, the sum of all
zero Bernoulli trials equals zero.'' After each iteration, we can
unify monomials having the same exponent, leading to a total of at
most $k+1$ monomials after each iteration. This unification step
allows to remove the combinatorial aspect of the problem, since
any monomial $x^i$ corresponds to a class of equivalent worlds,
such that this class contains only and all of the worlds where the
sum $\sum_{k=1}^N X_k=1$. In each iteration, the number of these
classes is $k$ and the probability of each class is given by the
coefficient of $x^i$.

An example showcasing the generating functions technique is given
in the following. This examples uses the identical Bernoulli
random variables used in Example \ref{ex:pbr123}.
\vspace{-0.1cm}
\begin{example}
Again, let $N=4$ and let $p_1=0.1$, $p_2=0.2$, $p_3=0.3$ and
$p_4=0.4$. We obtain the four generating polynomials
$poly(X_1)=(0.1x+0.9)$, $poly(X_2)=(0.2x+0.8)$,
$poly(X_3)=(0.3x+0.7)$, and $poly(X_4)=(0.4x+0.6)$. We trivially
obtain $\mathcal{F}^0=1$. Using Equation \ref{eq:genfkts2} we get\vspace{-0.1cm}
$$
\mathcal{F}^1=\mathcal{F}^{0}\cdot poly(X_1)=1\cdot (0.1x +
0.9)=0.1x + 0.9.
$$\vspace{-0.1cm}
Semantically, this polynomial implies that out of the first one
Bernoulli variables, the probability of having a sum of one is
$0.1$ (according to monomial $0.1x=0.1x^1$, and the probability of
having a sum of zero is $0.9$ (according to monomial $0.9=0.9x^0$.
Next, we compute $F^2$, again using Equation \ref{eq:genfkts2}:\vspace{-0.1cm}
$$
\mathcal{F}^2=\mathcal{F}^{1}\cdot poly(X_2)=(0.1x^1 +
0.9x^0)\cdot
(0.2x^1+0.8x^0)
=
$$
$$
0.02x^1x^1+0.08x^1x^0+0.18x^0x^1+0.72x^0x^0
$$\vspace{-0.1cm}
In this expansion, the monomials have deliberately not been
unified to give an intuition of how the generating function
techniques is able to identify and unify equivalent worlds. In the
above expansion, there is one monomial for each possible world.
For example, the monomial $0.18x^0x^1$ represents the world where
the first trial was unsuccessful (represented by the zero of the
first exponent) and the second trial was succesful (represented by
the one of the second exponent). The above notation allows to
identify the sequence of successful and unsuccessful Bernouli
trials, clearly leading to a total of $2^k$ possible worlds for
$\mathcal{F}^k$. However, we know that we only need to compute the
total number of successful trials, we do not need to know the
sequence of successful trials. Thus, we need to treat worlds
having the same number of successful Bernoulli trials
equivalently, to avoid the enumeration of an exponential number of
sequences. This is done implicitly by polynomial multiplication,
exploiting that\vspace{-0.1cm}
$$
0.02x^1x^1+0.08x^1x^0+0.18x^0x^1+0.72x^0x^0=0.02x^2+0.08x^1+0.18x^1+0.72x^0
$$\vspace{-0.1cm}
This representation no longer allows to distinguish the sequence
of successful Bernouli trials. This loss of information is
beneficial, as it allows to unify possible worlds having the same
sum of Bernoulli trials.
$$
0.02x^2+0.08x^1+0.18x^1+0.72x^0=0.02x^2+0.26x^1+0.72x^0
$$
The remaining monomials represent an equivalence class of possible
worlds. For example, monomial $0.26x^1$ represents all worlds
having a total of one successful Bernoulli trials. This is
evident, since the coefficient of this monomial was derived from
the sum of both worlds having a total of one successful Bernoulli
trials.In the next iteration, we compute:
$$
\mathcal{F}^3=\mathcal{F}^{2}\cdot
poly(X_3)=(0.02x^2+0.26x^1+0.72x^0)\cdot (0.3x+0.7)
$$
$$
=0.006x^2x^1+0.014x^2x^0+0.078x^1x^1+0.182x^1x^0+0.216x^0x^1+0.504x^0x^0
$$
This polynomial represents the three classes of possible worlds in
$\mathcal{F}^2$ combined with the two possible results of the
third Bernoulli trial, yielding a total of $3\dot2$ monomials.
Unification yields
$$
0.006x^2x^1+0.014x^2x^0+0.078x^1x^1+0.182x^1x^0+0.216x^0x^1+0.504x^0x^0=
$$
$$
0.006x^3+0.092x^2+0.398x^1+0.504
$$
The final generating function is given by
$$
\mathcal{F}^4=\mathcal{F}^{3}\cdot poly(X_4)=$$
$$(0.006x^3+0.092x^2+0.398x^1+0.504)\cdot (0.4x+0.6)=
$$
$$
.0024x^4+.0036x^3+.0368x^3+.0552x^2+.1592x^2+.2388x^1+.2016x^1+.3024x^0
$$
$$
=0.0024x^4+0.0404x^3+0.2144x^2+0.4404x+0.3024
$$
\end{example}
This polynomial describes the PDF of $\sum_{i=1}^4 X_i$, since
each monomial $c_ix^i$ implies that the probability, that out of
all four Bernoulli trials, the total number of successful events
equals $i$, is $c_i$. Thus, we get $P(\sum_{i=1}^4 X_i=0)=0.0024$,
$P(\sum_{i=1}^4 X_i=1)=0.0404$, $P(\sum_{i=1}^4 X_i=2)=0.2144$,
$P(\sum_{i=1}^4 X_i=3)=0.4404$ and $P(\sum_{i=1}^4 X_i=4)=0.3024$.
Note that this result equals the result we obtained by using the
Poisson binomial recurrence in the previous section.

\subsubsection*{Complexity Analysis}
The generating function technique requires a total of $N$
iterations. In each iteration $1\leq k \leq N$, a polynomial of
degree $k$, and thus of maximum length $k+1$, is multiplied with a
polynomial of degree $1$, thus having a length of $2$. This
requires to compute a total of $(k+1)\cdot 2$ monomials in each
iteration, each requiring a scalar multiplication. Thus leads to a
total time complexity of $\sum_{i=1}^N 2k+2\in O(N^2)$ for the
polynomial expansions. Unification of a polynomial of length $k$
can be done in $O(k)$ time, exploiting that the polynomials are
sorted by the exponent after expansion. Unification at each
iteration leads to a $O(n^2)$ complexity for the unification step.
This results in a total complexity of $O(n^2)$, similar to the
Poisson binomial recurrence approach.

An advantage of the generating function approach is that this
naive polynomial multiplication can be accelerated using Discrete
Fourier Transform (DFT). This technique allows to reduce to total
complexity of computing the sum of $N$ Bernoulli random variables
to $O(N log^2 N)$ (\cite{LiSahDes11}). This acceleration is
achieved by exploiting that DFT allows to expand two polynomials
of size $k$ in $O(k log k)$ time. Equi-sized polynomials are
obtained in the approach of \cite{LiSahDes11}, by using a divide
and conquer approach, that iteratively divides the set of $N$
Bernoulli trials into two equi-sized sets. Their recursive
algorithm then combines these results by performing a polynomial
multiplication of the generating polynomials of each set. More
details of this algorithm can be found in \cite{LiSahDes11}.

\section{Advanced Techniques for Managing Uncertain Spatial Data}\label{ref:overview}

\begin{table*}[t]
\centering
\caption{Advanced Topics in Querying and Mining Uncertain Spatial Data.\vspace{-0.1cm}}
\label{table:related_work}
\begin{tabular}{|l|l|c|c|c|c|}
\hline
\textbf{Topic}	& \textbf{Related Work} \\ \hline
Nearest Neighbor Query Processing  & \cite{ReyKalPra04,KriKunRen07,ChengCMC08,IjiIsh09,zhang2013voronoi,niedermayer2013probabilistic,schmid2017uncertain}  \\ \hline
$k$-Nearest Neighbor ($k$NN) Query Processing & \cite{kolahdouzan2004voronoi,BesSolIly08,CheCheCheXie09,bernecker2011novel} \\ \hline
Top-$k$ Query Processing & \cite{rds-07,Soliman07,YiLiKolSri08a} \\ \hline

Ranking of Uncertain Spatial Data & \cite{cl-08,BerKriRen08,CorLiYi09,SolIly09,LiSahDes09,bernecker2010scalable,dai2005probabilistic,LianChen09,bernecker2012probabilistic,hua2008ranking}  \\ \hline
Reverse $k$NN Query Processing  &\cite{LiaChe09,CheLinWanZhaPei10,bernecker2011efficient,emrich2014reverse}  \\ \hline
Skyline Query Processing & \cite{pei2007probabilistic,lian2008monochromatic,vu2013efficient,ding2014probabilistic,yang2018top} \\ \hline
Indexing Uncertain Spatial Data  & \cite{zhang2009effectively,emrich2012indexing,agarwal2009indexing}  \\ \hline
Maximum Range-Sum Query Processing & \cite{agarwal2018range,nakayama2017probabilistic,liu2019probabilistic}  \\ \hline
Querying Uncertain Trajectory Data & \cite{emrich2012querying,niedermayer2013similarity,zheng2011probabilistic} \\ \hline
Clustering Uncertain Spatial Data  &\cite{schubert2015framework,zufle2014representative,ngai2006efficient,kriegel2005density} \\ \hline
Frequent Itemset and Colocation Mining & \cite{bernecker2013model,wang2012efficient,bernecker2009probabilistic,bernecker2012probabilistic,wang2011finding} \\ \hline
\end{tabular}
\vspace{-0.2cm}
\end{table*}

The Paradigm of Equivalent worlds has been successfully applied to efficiently support many spatial query predicates and spatial data mining tasks. These more advanced techniques are out of scope of this book chapter, but the techniques presented in this chapter should help the interested reader to dive deeper into understanding state-of-the-art solutions, and to help the reader to contribute to this field.
An overview of research directions on uncertain spatial is provided in Table~\ref{table:related_work}.

Efficient solutions on uncertain data have been presented for ($1$)-nearest neighbor ($1$NN) queries~\cite{ReyKalPra04,KriKunRen07,ChengCMC08,IjiIsh09,zhang2013voronoi,niedermayer2013probabilistic,schmid2017uncertain}. The case of $1NN$ is special, as for $1NN$ the cases of object-based and result-based probabilistic result semantics are equivalent: Since a $1NN$ query only results a single result object. Thus, the probability of any object to be part of the result is equal the probability of this object to be the (whole) result. 
For $k$ Nearest Neighbor queries, this is not the case, as initially motivated in Figure~\ref{fig:toy}. 
For object-based result semantics (as explained in Section~\ref{subsec:ProbAnswerSem}), polynomial time solutions leveraging the paradigm of equivalent worlds have been proposed~\cite{bernecker2011novel}. For result-based result semantics, where each of the (potentially exponential many in $k$) results is associated with a probability, solutions have been presented in~\cite{BesSolIly08,CheCheCheXie09}. 

A related problem is Top-$k$ query processing which returns the $k$ best result objects for a given score function~\cite{rds-07,Soliman07,YiLiKolSri08a}. While these solution are not proposed in the context of spatial or spatio-temporal data, they are mentioned here as they can be applied to spatial data. For example, if the score function is defined as the distance to query object, this problem becomes equivalent to $k$NN. Solutions for result-based probabilistic result semantics are proposed in~\cite{Soliman07,rds-07} and for object-based result semantics in~\cite{YiLiKolSri08a}.

Another problem generalization are ranking queries, which return the Top-$k$ result ordered by score. For uncertain data using object-based result semantics, this yields a probabilistic mapping of each database mapping to each rank for the case of object-based result semantics. For example, it may return that object $o_1$ has a $80\%$ probability to be Rank 1, and a $20\%$ probability to be Rank 2. In the case of result-based probabilistic result semantics, each possible ranking of objects is mapped to a probability, for example, the ranking $[o_1,o_3,o_2]$ may have a $10\%$ probability.
Solutions for the result-based probabilistic result semantic case have been proposed in~\cite{SolIly09} having exponential run-time due to the hard nature of this problem. For the case of object-based probabilistic result semantics, first solutions having exponential run-time were proposed~\cite{BerKriRen08,cl-08}. Applying the paradigm of equivalent worlds, a number of solutions have been proposed concurrently and independently to achieve polynomial run-time (linear in the number of database objects times the number of ranks). The generating functions technique (as explained in Section~\ref{chap:sumofindependent}) was proposed for this purpose by Li et al.~\cite{LiSahDes09}. An equivalent approach using a technique called Poisson-Binomial Recurrence was simultaneously proposed by~\cite{bernecker2010scalable,hua2008ranking}. A comparison of the generating functions technique and the Poisson Binomial Recurrence, along with a proof of equivalence, can be found in~\cite{zufle2013similarity}. Other works shown in Table~\ref{table:related_work} include solutions for the case of existential uncertainty~\cite{dai2005probabilistic}, inverse ranking~\cite{LianChen09}, and spatially extended objects~\cite{bernecker2012probabilistic}, and the computation of the expected rank of an object.~\cite{CorLiYi09}.
Solution for indexing of uncertain spatial~\cite{agarwal2009indexing,chen2017indexing} and spatio-temporal~\cite{zhang2009effectively,emrich2012indexing} data have been proposed to speed up various of the previously mentioned query types.  

The problem of finding reverse $k$ nearest neighbors (R$k$NNs) have been studied for spatial data~\cite{LiaChe09,CheLinWanZhaPei10,bernecker2011efficient,CheLinWanZhaPei10} and spatio-temporal data~\cite{emrich2014reverse}. Solutions for skyline queries on uncertain data have been proposed in~\cite{pei2007probabilistic,lian2008monochromatic,vu2013efficient,ding2014probabilistic,yang2018top}.
More recently, the problem of answering Maximum Range-Sum Queries has been studied for uncertain data~\cite{agarwal2018range,nakayama2017probabilistic,liu2019probabilistic}.

Solutions tailored towards uncertain spatio-temporal trajectories, in which the exact location of an object at each point in time is a random variable have been proposed~\cite{emrich2012querying,niedermayer2013similarity,zheng2011probabilistic}. In this work, the challenge is to leverage stochastic processes that consider temporal dependencies. Such dependencies describe that the location of an object at a time $t$ depends on its location at time $t-1$.

Solutions for clustering uncertain data have been proposed~\cite{schubert2015framework,zufle2014representative,ngai2006efficient,kriegel2005density}. The challenge of clustering uncertain data is that the membership likelihood of on uncertain object to a cluster depends on other objects, making it hard to identify groups of worlds that are guaranteed to yields the same clustering result.

Finally, solutions for frequent itemset mining have been proposed for uncertain data~\cite{bernecker2013model,wang2012efficient,bernecker2009probabilistic,bernecker2012probabilistic}. While frequent itemset mining is not a spatial problem, it has applications in spatial co-location mining~\cite{wang2011finding,chan2019fraction}.

 Yet, many other spatial query predicates, as well as other probabilistic query predicates using different probabilistic result semantics are still open to study. The authors hopes that this chapter provides interested scholars with a starting point to fully understand preliminaries and assumptions made by existing work, as well as a general paradigm to develop efficient solutions for future work leveraging the Paradigm of Equivalent Worlds presented herein.

\vspace{-0.2cm}
\section{Summary}\label{sec:summary}
This chapter provided an overview of uncertain spatial data models and the concept of possible world semantics to interpret queries on these models. To understand the landscape of existing query processing algorithms on uncertain data, this chapter further surveyed different probabilistic result semantics and different probabilistic query predicates. To give the interested reader a start into this field, this chapter presented a general paradigm to efficiently query uncertain data based on the Paradigm of Equivalent Worlds, which aims at finding possible worlds that are guaranteed to have the same query result. As a case-study to apply this paradigm, this chapter provided solutions to efficiently compute range queries on uncertain data using an efficient recursion approach, as well as leveraging the concept of generating functions. 

Given this survey on modeling and querying uncertain spatial data, this chapter further provided a brief (and not exhaustive) overview of some research directions on uncertain spatial data. Many queries on uncertain data have already been solved efficiently, but many new challenges arise. For instance, only limited work has focused on streaming uncertain data, that is, handling uncertain data that changes rapidly. Another mostly open research direction is uncertain data processing in resources-limited scenarios such as edge computing. The author hopes that readers will find this overview useful to help readers understanding existing solutions and support readers towards adding their own research to this field. 